\documentclass[10pt, twocolumn]{IEEEtran}

\ifCLASSINFOpdf
  \usepackage[pdftex]{graphicx}
\else
\fi

\usepackage{amsfonts,amsmath,amsthm,amssymb,graphicx,algorithm,epstopdf,cite,url,subfigure}
\usepackage{psfrag,amssymb,amsmath,pifont,cite,graphics,graphicx,epsfig,subfigure,amscd,helvet,multirow,enumerate}
\usepackage{lmodern}
\usepackage{color}
\usepackage{epstopdf}
\usepackage{accents}
\usepackage{amssymb,amsmath,amscd,cite,enumerate}
\usepackage{graphics,graphicx,epsfig,psfrag,subfigure,threeparttable,url}
\usepackage[colorlinks,citecolor=black,linkcolor=black]{hyperref}
\usepackage{algpseudocode}
\usepackage[english]{babel}
\newtheorem{theorem}{Theorem}
\newtheorem{lemma}{Lemma}

\newtheorem{proposition}{Proposition}
\newtheorem{corollary}{Corollary}
\newtheorem{remark}{Remark}
\newtheorem{example}{Example}

\begin{document}
%
\title{Nearly Optimal Bounds for Orthogonal Least Squares}


\author{
Jinming~Wen$^\dag$, Jian Wang$^\dag$ and Qinyu Zhang~\IEEEmembership{Senior Member,~IEEE}
\thanks{J.~Wen is with  the Department of Electrical and Computer Engineering, University of Alberta, Edmonton T6G 2V4, Canada (e-mail: jinming1@ualberta.ca). 

J.~Wang is with the School of Data Science, Fudan University, Shanghai 200433, China. He was with the the Institute of New Media and Communications, Seoul National University, Seoul 151-742, Korea (Corresponding author, e-mail: wangjianeee@gmail.com).

Q.~Zhang is with the School of Electronic and Information Engineering, Harbin Institute of Technology, Shenzhen 518055, China (e-mail: zqy@hit.edu.cn). 
}

\thanks{$^\dag$These authors contributed equally to this work. }
}

\markboth{}
{Shell \MakeLowercase{\textit{et al.}}: Bare Demo of IEEEtran.cls for Journals}
%




\IEEEtitleabstractindextext{%
\begin{abstract}
In this paper, we study the orthogonal least squares (OLS) algorithm for sparse recovery.
On the one hand, we show  that if the sampling matrix $\mathbf{A}$ satisfies
the restricted isometry property (RIP) of order $K + 1$ with isometry constant
$$
\delta_{K + 1} < \frac{1}{\sqrt{K+1}},
$$
then OLS exactly recovers the support of any $K$-sparse vector $\mathbf{x}$
from its samples $\mathbf{y} = \mathbf{A} \mathbf{x}$ in $K$ iterations.
On the other hand, we show that OLS may not be able to recover the support of a $K$-sparse vector $\mathbf{x}$ in $K$ iterations for some $K$ if
$$
\delta_{K + 1} \geq \frac{1}{\sqrt{K+\frac{1}{4}}}.
$$
\end{abstract}

}
\maketitle
 

\IEEEdisplaynontitleabstractindextext

%
\IEEEpeerreviewmaketitle

\section{Introduction}
\label{sec:intro}
\IEEEPARstart{O}{r}thogonal least squares (OLS) is a classical greedy algorithm for subset selection in sparse approximation
and has attracted much attention in sparse recovery~\cite{chen1989orthogonal,rebollo2002optimized,foucart2013stability,herzet2012exact,wang2014recovery,maung2015improved}.
Consider the linear sampling model
\begin{equation}
\label{e:model}
\mathbf{y} = \mathbf{A} \mathbf{x},
\end{equation}
where $\mathbf{x} \in \mathbb{R}^n$ is a $K$-sparse vector (it has at most $K$ nonzero entries)
and $\mathbf{A} \in \mathbb{R}^{m \times n}$ is a sampling matrix.
The goal of sparse recovery is to identify the support of $\mathbf{x}$ (i.e., {\color{black}{the set of the positions of its nonzero elements}}) from the samples $\mathbf{y}$.
The OLS algorithm performs in an iterative {\color{black}{manner}}. In each iteration, it adds to {\color{black}{the estimated support}} an index which
leads to the maximum reduction {\color{black}{of}} the residual power.
The vestige of the active list is then eliminated from $\mathbf{y}$, yielding a residual update for the next iteration. See Table~\ref{a:OLS} for a mathematical description of OLS.
It has been shown that under appropriate conditions on $\mathbf{A}$, OLS yields exact recovery of the support of $\mathbf{x}$~\cite{foucart2013stability,herzet2012exact,wang2014recovery}.

In the sparse  approximation and sparse recovery literature, one of the typical methods that are closely related to OLS is the orthogonal matching pursuit (OMP) algorithm~\cite{pati1993orthogonal}.
The main difference between OLS and OMP lies in their greedy rules of updating the {\color{black}{estimated}} support in  each iteration.
While OLS seeks a candidate which results in the most significant decrease in the residual power,
OMP chooses a column that is most strongly correlated with the signal residual.
Consequently, the OLS and OMP algorithms coincide for the first iteration but usually differ afterward (see Section~\ref{sec:observation} for the justification). It has been empirically observed that OLS is computationally {\color{black}{more expensive yet is more reliable than OMP~\cite{herzet2012exact}.}} For more details on the differences between these two algorithms, see~\cite{blumensath2007difference} and {{\color{black}{the references therein}}.


In analyzing sparse recovery algorithms, {\color{black}{the restricted isometry property
(RIP) has been widely employed (see, e.g., \cite{candes2005decoding,foucart2009sparsest,needell2009cosamp,dai2009subspace,davenport2010analysis,liu2012orthogonal})}}.
A matrix $\mathbf{A}$ is said to satisfy the RIP of
order $K$ if there exists a constant $\delta \in (0, 1)$ such that~{\cite{candes2005decoding}}
\begin{equation}
  \label{eq:RIP} (1 - \delta) \| \mathbf{x} \|_2^2 \leq \| \mathbf{A x}
  \|_2^2 \leq (1 + \delta) \| \mathbf{x} \|_2^2
\end{equation}
for all $K$-sparse vectors $\mathbf{x}$. Specifically, the minimum of all
constants $\delta$'s satisfying~(\ref{eq:RIP}) is called the isometry constant and denoted
{\color{black}{by}} $\delta_K(\mathbf{A})$.
In the sequel, if there is no risk of confusion, we use $\delta_K$ instead of $\delta_K(\mathbf{A})$ for brevity.
In this paper, we utilize the RIP to study the recovery performance of OLS.
Our main goal is to develop a condition guaranteeing exact recovery of the support of $\mathbf{x}$ with the OLS algorithm.
In particular, our result is formally described in the following theorem.

\begin{theorem}
  \label{thm:1} 
Let $\mathbf{A} \in \mathbb{R}^{m \times n}$ be a sampling matrix with unit $\ell_2$-norm columns and satisfy the RIP with 
  \begin{equation}
    \label{eq:o}
     \delta_{K + 1} < \frac{1}{\sqrt{K + 1}}.
  \end{equation}
Then OLS exactly recovers the supports of all  $K$-sparse vectors $\mathbf{x} \in \mathbb{R}^n$
from the samples $\mathbf{y} = \mathbf{A} \mathbf{x}$ in $K$ iterations.
\end{theorem}




\setlength{\arrayrulewidth}{1.6pt}
\begin{table}
  \centering
\caption{The OLS Algorithm} \label{a:OLS}
\vspace{-2mm}
\begin{tabular}{@{}ll}
\hline \\ \vspace{-12pt} \\
\textbf{~~~Input}       &$\mathbf{A}$, $\mathbf{y}$, and sparsity level $K$.~ \\
\textbf{~~~Initialize}  & iteration counter $k = 0$, \\
                     & estimated support ${T}^{0} = \emptyset$, \\
                     & and residual vector $\mathbf{r}^{0} = \mathbf{y}$.
                     \\
\textbf{~~~While}       & $k < K$, \textbf{do}\\
                     & $k = k + 1$; \\
                     & Identify{\color{black}$^a$} \hspace{1.2mm}$t^{k} = \underset{i \in \{1, \cdots, n\}}{\arg \min}  \| \mathbf{P}^{\bot}_{{T}^{k - 1} \cup \{i\}} \mathbf{y} \|_{2}^{2}$; \\
                     & Augment \hspace{0mm}${T}^{k} = {T}^{k - 1} \cup \{ t^{k} \}$; \\
                     & Estimate \hspace{.6mm}$\mathbf{x}^{k} = \underset{\mathbf{u}:\text{supp}(\mathbf{u}) = {T}^k}{\arg \min} \|\mathbf{y} - \mathbf{A} \mathbf{u}\|_2$; \\
                     & Update \hspace{3mm}$\mathbf{r}^{k} = \mathbf{y} - \mathbf{A} \mathbf{x}^{k}$. \\
\textbf{~~~End}         \\
\textbf{~~~Output}  &$T^k$ and $\mathbf{x}^k$. \\
\vspace{-3.5pt} \\
\hline
\end{tabular}
    \begin{tablenotes}
        \item[a] \hspace{-4.5mm}{\color{black}$^a$}If the minimum occurs for multiple indices, {\color{black}{break the tie deterministically in favor of the first one.}}
    \end{tablenotes} \vspace{-3mm}
\end{table}
\setlength{\arrayrulewidth}{1.6pt}

Theorem~\ref{thm:1} improves~\cite[Theorem 1]{wang2014recovery} which shows that OLS performs the exact support recovery under
\label{footnote:aaa}}
\begin{equation}
\delta_{K + 1} < \frac{1}{\sqrt{K} + 2}. \label{eq:wangli}
\end{equation}
One can interpret from \eqref{eq:o} and \eqref{eq:wangli} that exact recovery with OLS can be ensured when the isometry constant $\delta_{K + 1}$ is inversely proportional to $\sqrt{K}$. In fact,  by exploring similarities between OLS and OMP, it can further be shown that the scaling law for $\delta_{K + 1}$ is necessary as well. Specifically, there exist counterexamples of $\mathbf{A}$ with unit $\ell_2$-norm columns and isometry constant~\cite{mo2012remarks,wang2012Recovery}
\begin{equation} \label{eq:nec}
\delta_{K + 1} = \frac{1}{\sqrt{K}},
\end{equation}
for which OMP fails to identify a support index of some $K$-sparse signals in the first iteration. Since OLS coincides with OMP for the first iteration, these counterexamples naturally apply to OLS, which implies that $\delta_{K + 1} < {1}/{\sqrt{K}}$ is also a necessary condition for the OLS algorithm.\footnote{Condition $\delta_{K + 1} < {1}/{\sqrt{K}}$ being necessary for OLS has also been shown in~\cite{herzet2012exact} for a sampling matrix $\mathbf{A}$ with non-unit $\ell_2$-norm columns.}
 The following result gives an improvement over this condition.


%


\begin{theorem}\label{thm:nece} 
  There exist a vector $\mathbf{x}  \in \mathbb{R}^n$ with some sparsity $K$
  and a sampling matrix $\mathbf{A} \in \mathbb{R}^{m \times n}$ with unit $\ell_2$-norm columns that satisfies the RIP with isometry constant
  \begin{equation}
    \label{eq:on} \delta_{K + 1} = \frac{1}{\sqrt{K + \frac{1}{4}}},
  \end{equation}
  such that OLS fails to recover the support of $\mathbf{x}$ {\color{black}{from the samples $\mathbf{y} = \mathbf{Ax}$}} in $K$ iterations.
\end{theorem}
One can notice that the gap between conditions~(\ref{eq:o}) and~(\ref{eq:on}) is
very small and vanishes for large $K$, which, therefore, indicates that condition (\ref{eq:o}) is nearly sharp for the OLS algorithm (see Figure~\ref{fig:comparison} for an illustration).

{\color{black}{
It is worth noting that, like in~\cite{foucart2013stability,herzet2012exact,wang2014recovery, chang2014improved},
conditions in Theorem~\ref{thm:1} and~\ref{thm:nece} rely on the assumption that $\mathbf{A}$ has unit $\ell_2$-norm columns.
In many applications, however, this assumption may not hold (e.g., when $\mathbf{A}$  is a Gaussian random matrix~\cite{candes2005decoding}), and one would need to build conditions for general matrices. Interestingly, by exploring the relationship between the RIP for general matrices and their normalized counterparts ({\bf Theorem~\ref{t:transferT}}), one can readily extend Theorem~\ref{thm:1} to the general cases ({\bf Corollary~\ref{thm:general}}).}}

%
%
%
%

The rest of this paper is organized as follows.
In Section~\ref{sec:pre}, we provide some observations and technical lemmas that are useful for our analysis. In Section~\ref{sec:analysis}, we prove Theorems~\ref{thm:1} and~\ref{thm:nece}.
Finally, we summarize and discuss our results in Section~\ref{sec:dis}.

\begin{figure}[t]
\centering
{\hspace{-2mm} \includegraphics[width = 86 mm]{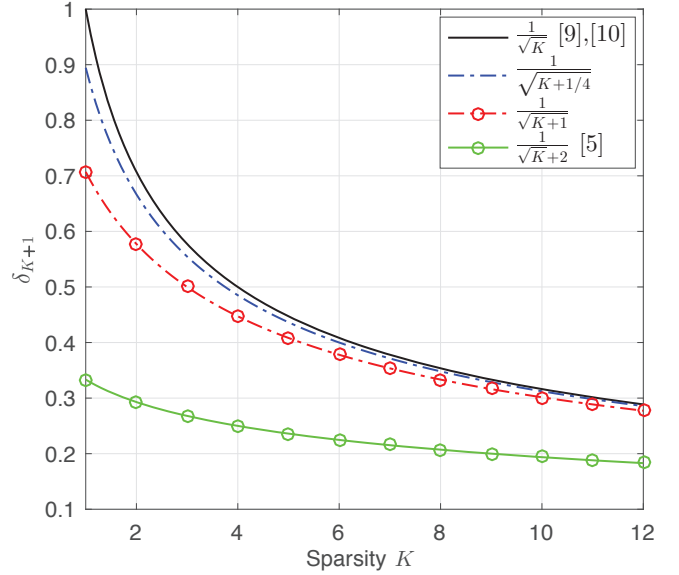}} \vspace{-3mm}
\caption{{\color{black}{An illustration of the upper bounds of $\delta_{K + 1}$.}}} \label{fig:comparison} \vspace{-3mm}
\end{figure}

\section{{\color{black}{Preliminaries}}} \label{sec:pre}
We first explain some notations that will be used throughout the {\color{black}{paper}}.
Let $T = \text{supp}(\mathbf{x}) = \{i|i \in \{1, \cdots, n\} \mbox{ such that }x_i \neq 0\}$
denote the support of vector $\mathbf{x}$. For a subset $S \subset \{1, \cdots, n\}$, let
$T \backslash S = \{i|i \in T \mbox{ but }i \notin S \}$.
Let $\mathbf{x}_{S} \in
\mathbb{R}^{| S |}$ be the restriction of the vector $\mathbf{x}$ to
the elements with indices in $S$. Similarly, let
$\mathbf{A}_{S} \in \mathbb{R}^{m \times | S |}$ be a
submatrix of $\mathbf{A}$ that contains only the columns indexed by $S$.
If $\mathbf{A}_{S}$ is full column rank, then
$\mathbf{A}_{S}^{\dagger} = (\mathbf{A}'_{S}
\mathbf{A}_{S})^{- 1} \mathbf{A}'_{S}$ is the
pseudoinverse of $\mathbf{A}_{S}$, where $\mathbf{A}'_{S}$ denotes the transpose of $\mathbf{A}_{S}$. $\mathbf{P}_{S} =
\mathbf{A}_{S} \mathbf{A}_{S}^{\dagger}$ stands for the
projection onto span$(\mathbf{A}_{S})$.
$\mathbf{P}_{S}^{\bot} = \mathbf{I} - \mathbf{P}_{S}$ is
the projection onto the orthogonal complement of
span$(\mathbf{A}_{S})$, where $\mathbf{I}$ is an identity matrix with $\mathbf{e}_j$ being its $j$-th column. 

\subsection{Observations} \label{sec:observation}
Before we proceed to the proof of Theorem~\ref{thm:1}, we give some useful
observations on the OLS algorithm. As detailed in
Table~\ref{a:OLS}, OLS selects in the ($k + 1$)-th ($0 \leq k < K$) iteration an index that results in the maximum reduction {\color{black}{of}} the residual power, {i.e.,}
\begin{equation}
  t^{k + 1} = \underset{i \in \{1, \cdots, n\}}{\arg \min} \|
  \mathbf{P}^{\bot}_{T^{k} \cup \{i\}} \mathbf{y} \|_2^2.
  \label{eq:identify}
\end{equation}
{\color{black}{By decomposing the projector $\mathbf{P}^{\bot}_{T^{k} \cup \{i\}}$ (see~{\cite{rebollo2002optimized,foucart2013stability,herzet2012exact,wang2014recovery,blumensath2007difference}} for details)}}, an alternative expression of (\ref{eq:identify}) can be given by
\begin{equation}
  t^{k + 1} = \underset{i \in \{1, \cdots, n\}}{\arg \max} \left| \left\langle \frac{\mathbf{P}^{\bot}_{T^{k}}
  \mathbf{A}_i}{\| \mathbf{P}^{\bot}_{T^{k}} \mathbf{A}_i \|_2},
  \mathbf{r}^{k} \right\rangle \right|, \label{eq:golsrule1}
\end{equation}
which offers a geometric interpretation of the selection rule of OLS. Specifically, the
columns of $\mathbf{A}$ are projected onto a subspace that is orthogonal to
the span of previously selected columns, and the normalized projected column that is most
strongly correlated with the current residual is chosen~\cite{herzet2012exact}.
Moreover, one can see from (\ref{eq:golsrule1}) that the behavior of OLS is unchanged by normalizing the columns of
$\mathbf{A}$ because ${\mathbf{P}^{\bot}_{T^{k}}
  \mathbf{A}_i}/{\| \mathbf{P}^{\bot}_{T^{k}} \mathbf{A}_i \|_2}$ would stay the same. Thus, for analytical convenience, we assume throughout the paper that
$\mathbf{A}$ has unit $\ell_2$-norm columns, i.e.,
 \begin{equation}
\|\mathbf{A}_i\|_2 = 1, ~ \text{for}~i = 1, \cdots, n. \label{eq:unitnorm}
\end{equation}
After the support list is updated (i.e., ${T}^{k + 1} = {T}^{k} \cup \{ t^{k + 1} \}$), OLS re-estimates the coefficients of $\mathbf{x}$ over the new list ${T}^{k + 1}$
by solving a least squares problem, which yields\footnote{Note that $\mathbf{A}_{T^k}^\dag$ is well-defined since $\delta_{K+1} \in (0, 1)$
{\color{black}{ensures that $k$ ($ \leq K + 1$) arbitrary columns of $\mathbf{A}$ are linearly independent.}}}
\[
 \mathbf{x}^{k+1}_{T^{k + 1}} = \mathbf{A}_{T^{k + 1}}^\dag \mathbf{y}~~\text{and}~~\mathbf{x}^{k + 1}_{\{1, \cdots, n\} \backslash T^{k + 1}} = \mathbf{0}.
\]
The residual vector is then updated as
\begin{eqnarray}
\mathbf{r}^{k + 1} &\hspace{-2mm} = &\hspace{-2mm} \mathbf{y} - \mathbf{A} \mathbf{x}^{k + 1} ~=~ \mathbf{y} - \mathbf{A}_{T^{k + 1}} \mathbf{x}^{k + 1}_{T^{k + 1}}
\nonumber \\
&\hspace{-2mm} = &\hspace{-2mm} (\mathbf{I}- \mathbf{A}_{T^{k + 1}} \mathbf{A}_{T^{k + 1}}^\dag)\mathbf{y}
= \mathbf{P}^{\bot}_{T^{k + 1}} \mathbf{y}. \label{eq:rkex}
\end{eqnarray}

We now take an observation on \eqref{eq:golsrule1}. Noting that
\begin{equation}
  \mathbf{P}^{\bot}_{T^{k}} = (\mathbf{P}^{\bot}_{T^{k}})' =
  (\mathbf{P}^{\bot}_{T^{k}})^2, \label{eq:pbot}
\end{equation}
$\langle \mathbf{P}^{\bot}_{T^{k}} \mathbf{A}_i, \mathbf{r}^{k}
  \rangle$ can be rewritten as
\begin{align}
  \langle \mathbf{P}^{\bot}_{T^{k}} \mathbf{A}_i, \mathbf{r}^{k}\rangle
  &\hspace{1mm}= \langle \mathbf{A}_i, (\mathbf{P}^{\bot}_{T^{k}})'\mathbf{r}^k \rangle
  \overset{\eqref{eq:rkex}}{=} \langle \mathbf{A}_i, (\mathbf{P}^{\bot}_{T^{k}})'
  \mathbf{P}^{\bot}_{T^{k}} \mathbf{y} \rangle \nonumber \\
  &\overset{\eqref{eq:pbot}}{=}\langle \mathbf{A}_i, \mathbf{P}^{\bot}_{T^{k}} \mathbf{y} \rangle \overset{\eqref{eq:pbot}}{=}  \langle \mathbf{A}_i, \mathbf{r}^{k}\rangle,
\end{align}
which together with \eqref{eq:golsrule1} imply the following proposition.

\begin{proposition}
  \label{prop:p1} {\color{black}{Consider the system model in \eqref{e:model} and the OLS algorithm. Let $\mathbf{r}^k$ be the residual produced in the $k$-th ($0 \leq k < K$) iteration of OLS.}} Then, OLS selects in the ($k + 1$)-th iteration the index
  \begin{equation}
    t^{k + 1} = \underset{i \in \{1, \cdots, n\}}{\arg \max}  \frac{| \langle \mathbf{A}_i, \mathbf{r}^{k}
    \rangle |}{\| \mathbf{P}^{\bot}_{T^{k}} \mathbf{A}_i \|_2}.
    \label{eq:golsrule}
  \end{equation}
\end{proposition}

{\color{black}{This proposition is a special case of \cite[Proposition~1]{wang2014recovery}}} and is of vital importance in analyzing the recovery condition of OLS. The
identification rule of OLS is akin to the OMP rule. Note that in the ($k + 1$)-th
iteration, OMP picks an index corresponding to the column which is most
strongly correlated with the signal
residual~{\cite{pati1993orthogonal}, i.e.,
\begin{equation}
t^{k + 1} = \underset{i \in \{1, \cdots, n\}}{\arg \max} | \langle \mathbf{A}_i, \mathbf{r}^{k} \rangle |.
\label{eq:golsrule2}
\end{equation}
Clearly the OLS rule differs from \eqref{eq:golsrule2} only in that it has an extra normalization factor
(i.e., $\| \mathbf{P}^{\bot}_{T^{k}} \mathbf{A}_i \|_2$, see the denominator of \eqref{eq:golsrule}). The normalization factor does not affect the first iteration of OLS
because $T^0 = \emptyset$ leads to
 \begin{equation}
\| \mathbf{P}^{\bot}_{T^0} \mathbf{A}_i
\|_2 = \| \mathbf{A}_i \|_2 \overset{\eqref{eq:unitnorm}}{=} 1. \nonumber
\end{equation} 
For the subsequent iterations, however, it
does make a difference since
$
\| \mathbf{P}^{\bot}_{T^k} \mathbf{A}_i
\|_2 \leq \| \mathbf{A}_i \|_2 = 1,~\forall~k \geq 1. 
$
In fact, as will be seen later, this factor
makes the analysis of OLS different and more challenging than that of OMP.

\subsection{Lemmas}
The following lemmas are useful for our analysis. 

\begin{lemma}
  [{\cite[Lemma 1]{dai2009subspace}}]\label{l:monot}If a matrix satisfies the
  RIP of both orders $K_1$ and $K_2$ with $K_1 \leq K_2$, then $\delta_{K_1}
  \leq \delta_{K_2}$.
\end{lemma}


\begin{lemma}
  [{\cite[Lemma 1]{li2015sufficient}}]\label{l:orthogonalcomp}Let sets $S_1$ {\color{black}{and}}
  $S_2$ satisfy $|S_2 \backslash S_1 | \geq 1$ and let matrix $\mathbf{A}$ obey
  the RIP of order $|S_1 \cup S_2 |$. Then, for any vector $\mathbf{u} \in
  \mathbb{R}^n$ supported on $S_2 \backslash S_1$,
  \begin{equation}
    (1 - \delta_{|S_1 \cup S_2 |}) \| \mathbf{u} \|_2^2 \leq \|
    \mathbf{P}^{\bot}_{S_1} \mathbf{A} \mathbf{u} \|_2^2
    \leq (1 + \delta_{|S_1 \cup S_2 |}) \| \mathbf{u} \|_2^2. \nonumber
  \end{equation}
\end{lemma}


As an refinement of \cite[Lemma 3.2]{davenport2010analysis}, Lemma \ref{l:orthogonalcomp}
says that when the columns of matrix $\mathbf{A}$ are projected onto a subspace
that is orthogonal to the span of its partial columns,
the resultant matrix also obeys the RIP.

We next introduce the following new lemma which will play a crucial role in proving our main result.
\begin{lemma}
  \label{l:projld}Suppose that $S \subset \{1, \cdots, n\}$ and let
  $\mathbf{A}$ have unit $\ell_2$-norm columns and satisfy  the RIP of order $|S| + 1$. Then, for any $i \in \{1, \cdots, n\} \backslash S$,
  \begin{equation}
    \label{e:projld}
    \| \mathbf{P}^{\bot}_S \mathbf{A}_i \|_2 \geq \sqrt{1 -
    \delta_{|S| + 1}^2}.
  \end{equation}
\end{lemma}
\begin{proof}
See Appendix~\ref{ss:projld}.
\end{proof}

 Lemma~\ref{l:projld} is useful for bounding the normalization factor in the selection of OLS
 (i.e., $\| \mathbf{P}^{\bot}_{T^{k}} \mathbf{A}_i \|_2$ in \eqref{eq:golsrule}). In particular, it improves existing results~\cite{li2015sufficient,foucart2013stability}
\begin{eqnarray}
\| \mathbf{P}^{\bot}_S \mathbf{A}_i \|_2 &\hspace{-2mm} \geq &\hspace{-2mm} \sqrt{1 - \delta_{|S|+1}} \label{eq:libound}
\\
\| \mathbf{P}^{\bot}_S \mathbf{A}_i \|_2 &\hspace{-2mm} \geq &\hspace{-2mm} \sqrt{1 - \frac{\delta_{|S|+1}^2}{1- \delta_{|S|}}},\label{eq:wangbound}
\end{eqnarray}
{\color{black}{which were actually used to prove condition \eqref{eq:wangli} in~\cite{wang2014recovery}.}}
\begin{remark}
 The bound in Lemma~\ref{l:projld} is tight, as the equality of~\eqref{e:projld} is attainable (see Example~\ref{l:projldnecessary} below). 
 The tightness of~\eqref{e:projld} essentially accounts for the near-optimality of the bound in Theorem~\ref{thm:1}. As will be seen in Appendix~\ref{ss:mainright}, it allows to build a unified condition guaranteeing correct selection
at each iteration of the OLS algorithm. By contrast, neither bound~\eqref{eq:libound} nor~\eqref{eq:wangbound} suffices.
\end{remark}

\begin{example}[{\color{black}{Sharpness of~\eqref{e:projld}}}]
  \label{l:projldnecessary}
  Let~$\rho \in (0, 1)$, $S = \{1, \cdots, |S|\}$, and
  \begin{equation}
  \mathbf{A} =
    \left[ \begin{array}{cccc}
          &                      &         &        \rho         \\
          &    \mathbf{I}_{|S|}  &         &        0             \\
          &                      &         &        \vdots        \\
          &                      &         &        0             \\
     0    &    \cdots            &    0    &        \sqrt{1 - \rho^2}
    \end{array} \right]_{(|S|+1) \times (|S|+1)}.    \nonumber
  \end{equation}
  Then one can verify that  $\mathbf{A}'\mathbf{A}$ has eigenvalues
  $
  \lambda_1 = 1 - \rho$, $\lambda_2 = \cdots = \lambda_{|S|} = 1$, $\lambda_{|S| + 1} = 1 + \rho,
  $
  which, by definition of the RIP and \cite[Remark 1]{dai2009subspace}, implies that $\mathbf{A}$ satisfies \begin{eqnarray}
  \delta_{|S|+1} & \hspace{-2mm} = & \hspace{-2mm}  \max \left\{1 - \lambda_{\min} (\mathbf{A}'\mathbf{A}),~ \lambda_{\max} (\mathbf{A}'\mathbf{A}) - 1 \right \} \nonumber \\
  & \hspace{-2mm} = & \hspace{-2mm}  \max \left \{1 - \lambda_1,  \lambda_{|S| + 1} - 1 \right\}  = \rho, \nonumber
  \end{eqnarray}
where $\lambda_{\min} (\mathbf{A}'\mathbf{A})$ and $\lambda_{\max} (\mathbf{A}'\mathbf{A})$ denote the minimal and maximal eigenvalues of $\mathbf{A}'\mathbf{A}$, respectively.

  One can also verify that
\begin{equation}
  \mathbf{P}^{\bot}_{S}=\mathbf{I}_{|S| + 1} - \mathbf{A}_{S} \mathbf{A}^\dag_{S}
=   \left[ \begin{array}{cccc}
          &                  &         &        0             \\
          & \hspace{1mm} \mathbf{0}_{|S|}\vspace{1mm} &         &        \vdots       \\
          &                  &         &        0             \\
     0    &    \cdots        &    0    &        1
    \end{array} \right], \nonumber
\end{equation}
and hence
$
\big \| \mathbf{P}^{\bot}_S \mathbf{A}_{|S| + 1} \big\|_2 = \sqrt{1 - \rho^2} = \sqrt{1 - \delta_{|S| + 1}^2}.$

\end{example}

The last lemma gives a connection between the isometry constant and the lower bound of residual power of OLS.
\begin{lemma}
\label{l:mainlemma}
{\color{black}{Consider the system model in \eqref{e:model} and the OLS algorithm.}} For any constant $\alpha > 0$ and $k \in \{0, \cdots, K - 1\}$, if the sampling matrix $\mathbf{A}$ obeys the RIP with 
\begin{equation}
\label{inq:deltaalpha}
\delta_{K - k + 1}<\frac{1}{\sqrt{\alpha+1}},
 \end{equation}
then the residual $\mathbf{r}^{k}$ of OLS satisfies
\begin{equation}
\label{inq:mainalpha}
 \| \mathbf{r}^{k} \|_2^2 > \sqrt{\alpha} ~ \| \mathbf{x}_{T \backslash T^k} \|_2 |\langle \mathbf{A}_{j}, \mathbf{r}^{k} \rangle|,~{\color{black}{j \in \{1, \cdots, n\} \backslash T}}.
  \end{equation}
\end{lemma}
\begin{proof}
See Appendix~\ref{ss:mainlemma}.
\end{proof}

\begin{remark}  \label{re:4444}
Lemma~\ref{l:mainlemma} is essentially motivated by~\cite[Lemma II.2]{mo2015sharp} and \cite[Lemma 1]{wen2015sharp}, but the result is stronger. Specifically, the RIP condition (i.e., \eqref{inq:deltaalpha}) depends on ``adjustable'' parameter $\alpha$ which can be any positive number.
In contrast,~\cite{mo2015sharp,wen2015sharp} focused only on the case where $\alpha$ is a positive integer. {\color{black}{In fact, $\alpha$ taking non-integer values is precisely needed for proving Proposition~\ref{p:mainright}}}, which is a main step of the proof for Theorem~\ref{thm:1} (see Appendix~\ref{ss:mainright}). {\color{black}{Moreover, we would like to point out that the result in Lemma~\ref{l:mainlemma} works for both OMP and OLS; whereas, those of~\cite{mo2015sharp,wen2015sharp} are valid only for the OMP case.}}
\end{remark}


\section{Main Analysis} \label{sec:analysis}

\subsection{Connection to Existing Analyses} \label{sec:distinction}

The proof of Theorem~\ref{thm:1} is closely related to~\cite{liu2012orthogonal,mo2012remarks,wang2012Recovery,mo2015sharp}, in which recovery conditions for the OMP algorithm were established. However, there is a key distinction between the analyses of these works and our proof.
Before discussing the distinction, we first mention a common property of OMP and OLS, which says that if the previous $k$ selections of OMP/OLS ($0 \leq k < K$) are correct ($T^k \subset T$), then the residual $\mathbf{r}^k$ can be viewed as the samples of a $K$-sparse vector with matrix $\mathbf{A}$. More precisely,
\begin{equation}
\label{eq:romp}
\mathbf{r}^k = \mathbf{y} - \mathbf{A} \mathbf{x}^k = \mathbf{Az},
\end{equation}
where $\mathbf{z} = \mathbf{x} - \mathbf{x}^k$ is a $K$-sparse vector supported on $T$.\footnote{Note that $\mathbf{x}^k$ is supported on $T^k$ ($\subset T$) and $\mathbf{x}$ is supported on $T$.}
For the OMP algorithm, this property enables a recursive proof for the recovery condition, since the index selection at every iteration of OMP is based on the correlation between $\mathbf{r}^k$ and $\mathbf{A}$.
To be more concrete, since $\mathbf{z}$ is also $K$-sparse,
the condition for the first iteration of OMP ($k = 0$), which guarantees to select a support index of the $K$-sparse vector $\mathbf{x}$ from its samples $\mathbf{y} = \mathbf{A x}$, would naturally guarantee to select a support index of $\mathbf{z}$ from its ``samples'' $\mathbf{r}^k = \mathbf{Az}$.
Notably, the recursive proof substantially simplifies the analysis of OMP, as one only needs to consider the first iteration (see \cite{liu2012orthogonal, mo2012remarks, wang2012Recovery,wang2012Generalized,mo2015sharp}).

For the OLS case, however, \eqref{eq:romp} does not enable a recursive proof due to the normalization factor arised in its selection rule (see the right-hand side of~\eqref{eq:golsrule}). Specifically, while the index selection at the first iteration of OLS is based on the correlation between $\mathbf{y}$ and $\mathbf{A}$, that of the subsequent iterations would be based on the correlation between $\mathbf{r}^k$ and a different matrix, whose columns are composed of $\left\{{\mathbf{A}_i}/{\| \mathbf{P}^{\bot}_{T^{k}} \mathbf{A}_i \|_2}\right\}_{i = 1, \cdots, n}.$
As such, the condition for the first iteration of OLS does not apply immediately to its subsequent iterations. Therefore,  to build the recovery condition for the OLS algorithm, we need to consider the first iteration as well as the subsequent ones.

In analyzing the subsequent iterations of OLS, the difficulty lies in dealing with the normalization factors $\{\| \mathbf{P}^{\bot}_{T^{k}} \mathbf{A}_i \|_2\}_{i = 1, \cdots, n}$.
The primary novelties of our techniques are 
i) to incorporate simultaneously the normalization factor into the RIP condition {\color{black}{as well as into the estimator of}} $\| \mathbf{r}^{k} \|_2^2$, which is done by applying
 $\alpha = {\sqrt{|T \backslash T^k|}}/{\| \mathbf{P}^{\bot}_{T^{k}} \mathbf{A}_{i} \|_2}$ in Lemma~\ref{l:mainlemma},
and
ii) to bound ${\| \mathbf{P}^{\bot}_{T^{k}} \mathbf{A}_{i} \|_2}$ with a tight inequality established in Lemma~\ref{l:projld}. 
Interestingly, in this way we are able to obtain a unified condition for every iteration of the OLS algorithm.

The proof for Theorem~\ref{thm:1} was also motivated by~\cite{wang2014recovery}. There, OLS finding multiple atoms per iteration has been studied under the name of multiple OLS (MOLS).

\begin{theorem}[Wang and Li~\cite{wang2014recovery}] \label{thm:333}
Let $L  \leq \min \{K,
\frac{m}{K} \}$ be the number of indices picked at each iteration
of MOLS. Let $\mathbf{x} \in \mathcal{R}^n$ be any $K$-sparse signal and
$\mathbf{A} \in \mathcal{R}^{m \times n}$ be the sampling
matrix with unit $\ell_2$-norm columns. Then if $\mathbf{A}$ satisfies the RIP
  with
  \begin{subequations}
\begin{align}
&\delta_{K+1}  < \frac{1}{\sqrt{K} + 2},~~~~~L=1,
 \label{eq:mu1}\\
&\delta_{LK}  < \frac{\sqrt{L}}{\sqrt{K} + 2\sqrt{L}},~~L>1.
 \label{eq:mu2}  
\end{align}
\end{subequations}
  MOLS exactly recovers $\mathbf{x}$ from $\mathbf{y} = \mathbf{A x}$ within $K$ steps.
\end{theorem}

In this work, some newly developed technologies, such as Lemma~\ref{l:projld} and~\ref{l:mainlemma},  were used to advance \eqref{eq:mu1}. We mention that these lemmas may also be utilized to  refine~\eqref{eq:mu2}. For example, since Lemma~\ref{l:projld} offers a sharp version of~\eqref{eq:wangbound}, it can be used to derive a tighter result of~\cite[eq. E.6]{wang2014recovery}. 
 Also, Lemma~\ref{l:mainlemma} can be useful to better bound the residual power $\| \mathbf{r}^{k} \|_2$ of MOLS. However, as the analysis of~\cite{wang2014recovery} involves many approximations due to existence of $L$, we expect that the analysis optimizing the constant in~\eqref{eq:mu2} will be complicated. Whether it is possible to obtain a condition analogous to \eqref{eq:o} (e.g., $\delta_{LK} < \sqrt{\frac{{L}}{{K + L}}}$, $L > 1$) remains an interesting open question.


\subsection{Proof of Theorem~\ref{thm:1}}\label{ss:l2}

To prove Theorem~\ref{thm:1}, it suffices to show that the OLS algorithm selects
  a correct index in  every iteration. Our proof works by induction.
  Suppose that OLS makes correct selections in each of the previous $k$ iterations ($T^{k} \subset T$).
  Then we shall show that the algorithm also
  selects a correct index in the ($k + 1$)-th iteration, that is, $
  t^{k + 1} \in T\backslash T^k$.\footnote{Indices in $T^k$ cannot be re-selected in subsequent iterations of OLS because $\langle \mathbf{A}_i, \mathbf{r}^k \rangle_{i \in T^k} = 0$. See~\cite[Lemma 7]{wang2012Recovery} for more details.}
  Here, we assume that $0 \leq k < |T|$. Thus, the first selection of OLS corresponds to the case of
  $k = 0$, for which case our induction hypothesis $T^{k} \subset T$ still holds because $T^0 = \emptyset$.

%

Our proof relies on the following two propositions, which respectively characterize a lower bound for
$
\max_{i \in T \backslash T^k} {|\langle \mathbf{A}_i, \mathbf{r}^{k} \rangle|}/{\| \mathbf{P}^{\bot}_{T^{k}} \mathbf{A}_i \|_2}
$
and an upper bound for
$
\max_{i \in \{1, \cdots, n\} \backslash T} {|\langle \mathbf{A}_i, \mathbf{r}^{k} \rangle|}/{\| \mathbf{P}^{\bot}_{T^{k}} \mathbf{A}_i \|_2}.
$ Clearly if the former dominates the latter, then a reliable selection is ensured at the ($K + 1$)-th iteration of the OLS algorithm.

\begin{proposition}
  \label{p:mainleft}
Consider the system model in \eqref{e:model} and the OLS algorithm. Let $\mathbf{r}^k$ be the residual produced in the $k$-th ($0 \leq k < K$) iteration of OLS and suppose $T^k \subset T$. Then if $\mathbf{A}$ satisfies \eqref{eq:unitnorm} and the RIP with \eqref{eq:o}, OLS satisfies
  \begin{eqnarray}
      \label{e:mainleft}
    \max_{i \in T \backslash T^k} \frac{|\langle \mathbf{A}_i, \mathbf{r}^{k} \rangle|}{\| \mathbf{P}^{\bot}_{T^{k}} \mathbf{A}_i \|_2} &\hspace{-2mm}\geq&\hspace{-2mm} \frac{\| \mathbf{r}^{k} \|_2^2}{\sqrt{|T \backslash T^k|} \| \mathbf{x}_{T \backslash T^k} \|_2}, \\
    \label{e:mainright} \max_{i \in \{1, \cdots, n\} \backslash T} \frac{|\langle \mathbf{A}_i, \mathbf{r}^{k} \rangle|}{\| \mathbf{P}^{\bot}_{T^{k}} \mathbf{A}_i \|_2} &\hspace{-2mm}<&\hspace{-2mm} \frac{\| \mathbf{r}^{k} \|_2^2}{\sqrt{|T \backslash T^k|} \| \mathbf{x}_{T \backslash T^k} \|_2}.
  \end{eqnarray}
\end{proposition}
\begin{proof}
See Appendix~\ref{ss:mainleft}.
\end{proof}

With the foregoing proposition, we immediately obtain that under~\eqref{eq:o} and~\eqref{eq:unitnorm},
  \begin{equation}
    \label{e:l2cond} \max_{i \in T \backslash T^k} \frac{|\langle \mathbf{A}_i, \mathbf{r}^{k} \rangle|}{\| \mathbf{P}^{\bot}_{T^{k}} \mathbf{A}_i \|_2} > \max_{i \in \{1, \cdots, n\} \backslash T} \frac{|\langle \mathbf{A}_i, \mathbf{r}^{k} \rangle|}{\| \mathbf{P}^{\bot}_{T^{k}} \mathbf{A}_i \|_2}, \nonumber
  \end{equation}
which together with Proposition~\ref{prop:p1} implies that
\begin{equation}
t^{k + 1} = \underset{i \in \{1, \cdots, n\}}{\arg \max}  \frac{| \langle \mathbf{A}_i, \mathbf{r}^{k}
    \rangle |}{\| \mathbf{P}^{\bot}_{T^{k}} \mathbf{A}_i \|_2} \in T \backslash T^k. \nonumber
\end{equation}
In other words, the OLS algorithm selects a correct index
  in the ($k + 1$)-th iteration. This establishes the theorem.

\begin{remark}
\label{r:OMPOLS}
Propositions~\ref{p:mainleft} and~\ref{p:mainright} imply that when $k = 0$,
 \begin{equation}
    \label{e:mainleft2}
    \max_{i \in T} |\langle \mathbf{A}_i, \mathbf{y} \rangle| \geq
     \frac{\| \mathbf{y} \|_2^2}{\sqrt{K} \| \mathbf{x} \|_2} > \max_{i \in \{1, \cdots, n\} \backslash T} |\langle \mathbf{A}_i, \mathbf{y} \rangle|
\end{equation}
holds under~\eqref{eq:o}. Thus a correct selection is ensured at the first iteration of OLS. Since OMP coincides with OLS for the first iteration,~the RIP condition naturally applies to the OMP algorithm. Furthermore, by the recursive-proof argument of OMP in Section~\ref{sec:distinction}, it also guarantees correct selections in the subsequent iterations. {\color{black}{Therefore, the condition~\eqref{eq:o} in Theorem~\ref{thm:1} is also sufficent for the OMP recovery, which matches the recent work of Mo~\cite{mo2015sharp}.}}
\end{remark}


\subsection{Proof of Theorem~\ref{thm:nece}}\label{sec:example}

Note that if a wrong selection is made at the first iteration of OLS, then recovering the support of $\mathbf{x}$ in $K$ iterations is impossible.
Thus, to prove Theorem~\ref{thm:nece}, it suffices to show that there exist a sparse signal $\mathbf{x}$ with sparsity $K$
and a sampling matrix $\mathbf{A}$ satisfying \eqref{eq:unitnorm} and
$\label{eq:on55} \delta_{K + 1} = \frac{1}{\sqrt{K + \frac{1}{4}}},$
  for which OLS is unable to identify a support index of $\mathbf{x}$ in the first iteration. {\color{black}{One such example is given as follows (see Appendix~\ref{app:example} for details).

\begin{example}
\label{ex:countex}
  Consider $K = 2$ and let
\begin{equation}
      \mathbf{x} = \left[ \begin{array}{c}
       \hspace{-.5mm} 0 \hspace{-.5mm}\\
       \hspace{-.5mm} - 1 \hspace{-.5mm}\\
       \hspace{-.5mm} 1 \hspace{-.5mm}
    \end{array} \right] ~ \text{and} ~ ~
    \mathbf{A} =
    \left[ \begin{array}{ccc}
     1           & \frac{1}{3}                         & -\frac{1}{3} \\
     0 & \frac{2 \sqrt{2}}{3}      & \frac{\sqrt{2}}{3} \\
     0 &   0    & \frac{\sqrt{6}}{3}
    \end{array} \right]. \label{eq:AAAAA}
  \end{equation}
Then, $ \mathbf{A}$ satisfies \eqref{eq:unitnorm}, and moreover, we have
  \begin{equation}
  \label{eq:phi}
    \mathbf{y} = \mathbf{A x} = \left[ \begin{array}{c}
       \hspace{-.5mm} -\frac{2}{3} \hspace{-.5mm}\\
       \hspace{-.5mm} - \frac{\sqrt{2}}{3} \hspace{-.5mm}\\
       \hspace{-.5mm} \frac{\sqrt{6}}{3} \hspace{-.5mm}
    \end{array} \right] ~ \text{and} ~ ~ \mathbf{A}'\mathbf{A} =
    \left[ \begin{array}{ccc}
      1            & \frac{1}{3}   & - \frac{1}{3} \\
      \frac{1}{3}  &      1        & \frac{1}{3} \\
     - \frac{1}{3}  &  \frac{1}{3} & 1
    \end{array} \right]. \nonumber
  \end{equation}

    One can verify that the eigenvalues of matrix $\mathbf{A}'\mathbf{A}$ are $\lambda_1 = \lambda_2 = \frac{4}{3}~\text{and}~\lambda_3 = \frac{1}{3}.$
%
Then, by definition of the RIP and \cite[Remark 1]{dai2009subspace}, we have
  \begin{eqnarray}
   \delta_{K + 1} =  \max \{ \lambda_1 - 1, \lambda_2 - 1, 1 - \lambda_3 \} =   \frac{2}{3} = \frac{1}{\sqrt{K + \frac{1}{4}}}\bigg|_{K = 2}. \nonumber
  \end{eqnarray}

However, noting that $T^0 = \emptyset$ and 
\begin{equation}
  \frac{\left| \langle \mathbf{A}_1, \mathbf{y} \rangle \right|}{\| \mathbf{P}^{\bot}_{T^0} \mathbf{A}_1 \|_2} = \frac{\left| \langle \mathbf{A}_2, \mathbf{y} \rangle \right|}{\| \mathbf{P}^{\bot}_{T^0} \mathbf{A}_2 \|_2}  = \frac{\left| \langle \mathbf{A}_3, \mathbf{y} \rangle \right|}{\| \mathbf{P}^{\bot}_{T^0} \mathbf{A}_3 \|_2}  = \frac{2}{3}, \nonumber
  \end{equation}
by Table \ref{a:OLS}, the OLS algorithm  will choose a wrong index $t^1 =1$ in the first iteration.
As a result, OLS will fail to exactly recover the support of $ \mathbf{x}$ in $K$ ($= 2$) iterations.

  \end{example}
}}

\section{Discussions} \label{sec:dis}

Thus far, we have presented a nearly optimal sufficient condition that ensures the OLS algorithm
to exactly recover the support of sparse signals.
In this section, we discuss some issues that arise from our analysis.

Firstly, like in~\cite{foucart2013stability,herzet2012exact,wang2014recovery, chang2014improved},
our analysis relies on the assumption of unit $\ell_2$-norm columns for the sampling matrix $\mathbf{A}$.
In many practical scenarios, however, the columns of sampling matrices may not have unit $\ell_2$-norm
(e.g., {\color{black}{Gaussian random matrices~\cite{candes2005decoding}}}).
This naturally raises the question of whether similar bounds of OLS can be obtained for general matrices.
The following theorem aims to answer this question. Specifically, it characterizes the relationship between isometry constants for a general matrix
and its normalized counterpart, which together with Theorem~\ref{thm:1}
leads to a recovery bound of OLS for general matrices.

\begin{theorem}
  \label{t:transferT} 
Let $\hat{\mathbf{A}}$ be an $m \times n$ matrix satisfying the RIP of order $K$ with isometry constant $\delta_K(\hat{\mathbf{A}})$ and let $\mathbf{A}$ be the column-normalized form of $\hat{\mathbf{A}}$, i.e.,
$\mathbf{A} = \hat{\mathbf{A}} \mathbf{D}$, where $\mathbf{D}\in \mathbb{R}^{n\times n}$ is a diagonal matrix with
$d_{ii}=\frac{1}{\|\hat{\mathbf{A}}_i\|_2},~i = 1, \cdots, n.$
Then, $\mathbf{A}$ satisfies the RIP of order $K$ with $\delta_K(\mathbf{A}) \leq \gamma$,
\begin{eqnarray}
\label{e:gamma}
\text{where}~&& \hspace{-7mm} {\color{black}{ \gamma = \max   \left\{  \left[1   +   \delta_K(\hat{\mathbf{A}}) \right]   \max_{1\leq i\leq n}   d_{ii}   -   1, \right.}}\nonumber \\
&&   
~~~~~~~~~~~~~~~{\color{black}{ \left. 1   -   \left[1   -   \delta_K(\hat{\mathbf{A}}) \right]   \min_{1\leq i\leq n}   d_{ii}   \right\}. }}~~~~~~~
\end{eqnarray}
\end{theorem}
\begin{proof}
See Appendix~\ref{ss:relat}.
\end{proof}

\begin{remark}
{\color{black}{
\label{r:RIC}
Theorem~\ref{t:transferT} offers an upper bound on $\delta_K(\mathbf{A})$ in terms of $d_{ii}$ and $\delta_K(\hat{\mathbf{A}})$. To see the tightness of this bound, we take the following example. If $d_{ii}=1$ for $i=1, \cdots, n$, then we have $\mathbf{A}=\hat{\mathbf{A}}$ and consequently  $\delta_K(\mathbf{A})=\delta_K(\hat{\mathbf{A}})$, which means that the equality is attainable. }}

%
\end{remark}


{\color{black}{
From \eqref{e:gamma}, one can see that the upper bound of $\delta_K(\mathbf{A})$  depends on both
$\max\limits_{1\leq i\leq n}d_{ii}$ and $\min\limits_{1\leq i\leq n}d_{ii}$. To see an explicit connection between $\delta_K(\mathbf{A})$ and $\delta_K(\hat{\mathbf{A}})$, we provide the following corollary.}}

\begin{corollary}
  \label{c:transferT}
Under the assumption of Theorem \ref{t:transferT}, matrix $\mathbf{A}$ satisfies the RIP of order $K$ with isometry constant
  \begin{equation}
    \label{e:transferT2}
\delta_K(\mathbf{A}) \leq \frac{2\delta_K(\hat{\mathbf{A}})}{1-\delta_K(\hat{\mathbf{A}})}.
  \end{equation}
\end{corollary}

\begin{proof}
See Appendix~\ref{app:cc}.
\end{proof}

By applying Corollary~\ref{c:transferT} to Theorem~\ref{thm:1}, Corollary~\ref{thm:general} below is immediate.
We remark that Theorem~\ref{t:transferT} may also be useful for analyzing RIP conditions of other sparse recovery algorithms.
To be specific, if one wishes to derive a condition based on general matrices (whose columns are not necessarily normalized to be unitary),
one may, alternatively, build a condition by assuming normalized columns for analytical convenience,
and then transfer the condition to the general case by using Theorem~\ref{t:transferT}.

\begin{corollary}
  \label{thm:general} Let $\mathbf{x} \in \mathbb{R}^n$ be a $K$-sparse
  vector and $\hat{\mathbf{A}} \in \mathbb{R}^{m \times n}$ be the sampling matrix satisfying the RIP with
\begin{equation}
  \delta_{K + 1}(\hat{\mathbf{A}}) < \frac{1}{2\sqrt{K + 1}+1}.  \label{eq:general}
  \end{equation}
Then, OLS exactly recovers $\mathbf{x}$ from $\mathbf{y} = \hat{\mathbf{A}} \mathbf{x}$ in $K$ steps.
\end{corollary}

Secondly, although in this paper we are primarily interested in recovering $K$-sparse signals, our analysis can be possibly extended to the situations where the signals of interest have specific properties, such as having non-negative or exponentially decaying nonzero entries,  in addition to the sparsity nature.
In fact, recovering non-negative sparse signals arises in many application domains, where the signals have physical interpretations; see, e.g.,~\cite{zass2006nonnegative} and {\color{black}{the}} references therein.
%
%
We expect that the recovery condition for OLS can be improved if those specific properties are incorporated in the analysis. 
%
%
%
%
%
{\color{black}{Moreover, we would like to mention that, by following some techniques developed in~\cite{CaiW11, chang2014improved,determe2016exact,wang2014recovery,wen2015sharp}, our analysis may also be extended to the noisy case, in which rather than~\eqref{e:model}, we observe 
$
\mathbf{y} = \mathbf{A} \mathbf{x}+\mathbf{v}.
$
Here, $\mathbf{v}$ is the noise vector commonly assumed to be $\ell_2$-bounded (i.e., $\|\mathbf{v}\|_2\leq \epsilon$ for some constant $\epsilon$ \cite{DonET06}), $\ell_{\infty}$-bounded (i.e., $\|\mathbf{A}\mathbf{v}\|_{\infty}\leq \epsilon$ for some constant $\epsilon$~\cite{CaiW11}), or Gaussian  (i.e., $v_i\sim \mathcal{N}(0, \sigma^2)$~\cite{CanT07}).}}
%
By studying the behavior of OLS in those scenarios, one may gain a better insight of the recovery ability of this algorithm.


{\color{black}{
%
%
%

%
Thirdly, in Theorem~\ref{thm:1} and \ref{thm:nece}, we have, respectively, established a sufficient condition and a necessary condition for exact recovery of sparse signals via OLS (i.e., $\delta_{K + 1} < \frac{1}{\sqrt{K + 1}}$ and $\delta_{K + 1} < \frac{1}{\sqrt{K + \frac{1}{4}}}$).  One may notice that there remains a small gap between these two conditions. To bridge this gap, it may require {\color{black}{a refined analysis based on the counterexample}} in Section~\ref{sec:example}, Specifically, as detailed in Appendix~\ref{app:example}, our counterexample considers the $3$-dimensional case, in which we have only three variables (i.e., $a$, $b$ and $c$ in \eqref{eq:hi}) to tune. For the higher dimensional cases, there will be more unknowns to be optimized and thus better condition can be expected. In these cases, however, the problem will also become more complex (than that in Appendix~\ref{app:example}). Whether it is possible to close the gap and get a sharp condition for OLS in the higher dimensional cases is an interesting open question. 

%
%
%
%
%
%
}}


Finally, we would like to mention that while Theorem~\ref{thm:1} demonstrates a near-optimal condition for OLS when it iterates $K$ times,
there is still significant room for improving the result if OLS is allowed to perform more than $K$ iterations.
In fact, it has been shown that if OLS runs $12K$ iterations,
exact recovery is guaranteed when the isometry constant is an absolute constant independent of $K$~\cite{foucart2013stability}.
This offers  many benefits in the sampling complexity. For example, for Gaussian random sampling matrices,
it has been shown that the number $m$ of samples scales inversely to the square of isometry constants with probability exponentially close to one~\cite{candes2005decoding};
thus, an improved isometry constant directly {\color{black}{leads to a reduction of the sampling complexity}}.
However, it should be noted that executing more iterations is also associated with higher computational cost.
Meanwhile, selection of too many incorrect indices could significantly degrade the reconstruction performance, particularly when noise is present~\cite{ding2013perturbation}.
Therefore, finding an appropriate trade-off between the computational cost and sampling complexity for OLS can be of vital importance, 
and our future work will be directed towards investigating this issue.

\appendices
\numberwithin{equation}{section}
\newcounter{mytempthcnt}
\setcounter{mytempthcnt}{\value{theorem}}
\setcounter{theorem}{2}

\section{Proof of Lemma~\ref{l:projld}}\label{ss:projld}

\begin{proof}
  We first consider the case that $S = \emptyset$. In this case, we have $\mathbf{P}^{\bot}_S = \mathbf{I}$ and hence
\[
  \| \mathbf{P}^{\bot}_S \mathbf{A}_i \|_2 = \| \mathbf{A}_i \|_2 \overset{\eqref{eq:unitnorm}}{=} 1 \geq
     \sqrt{1 - \delta_{|S| + 1}^2}.
\]

Next, we consider the case that $S \neq \emptyset$. Observe that
$
\mathbf{A}_i = \mathbf{P}_S \mathbf{A}_i + \mathbf{P}^{\bot}_S \mathbf{A}_i
$
  Let $\theta \in (0, \frac{\pi}{2}]$ denote the angle between $\mathbf{A}_i$
  and $\mathbf{P}_S \mathbf{A}_i$. A geometric illustration of $\theta$ and $\mathbf{P}^{\bot}_S \mathbf{A}_i$ is given in Figure~\ref{fig:projection}. Then by~\eqref{eq:unitnorm}, we reach
  \begin{equation}
    \label{e:projeq} \| \mathbf{P}^{\bot}_S \mathbf{A}_i \|_2 = \|
    \mathbf{A}_i \|_2 \sin \theta \overset{\eqref{eq:unitnorm}}{=} \sin \theta.
  \end{equation}
   From the definition of $\mathbf{P}_S$, there exists a vector $\mathbf{z} \in
  \mathbb{R}^n$, which is supported on $S$, such that
  $
  \mathbf{P}_S \mathbf{A}_i = \mathbf{A}_S \mathbf{z}_S = \mathbf{Az}.
  $
%
%


Note that $i \notin S$ and $\mathbf{A}_i = \mathbf{A} \mathbf{e}_i$ where  $\mathbf{e}_i$ denotes the $i$-th column of the $n\times n$ matrix.
Applying \cite[Lemma 2.1]{chang2014improved}, which is an improvement of Cand{\`e}s and Tao's result~\cite{candes2005decoding} recently obtained by Chang and Wu, with $\mathbf{u} = \mathbf{e}_i$ and $\mathbf{v} = \mathbf{z}$
yields
%
  $
  \cos \theta \leq \delta_{\left|\text{supp}(\mathbf{z}) \hspace{.5mm} \cup \hspace{.5mm} \text{supp}(\mathbf{e}_i) \right|} = \delta_{|S| + 1},
  $
  and hence
    $
  \sin \theta = \sqrt{1 - \cos^2 \theta} \geq \sqrt{1 - \delta_{|S| + 1}^2},
  $
which together with~\eqref{e:projeq} concludes this lemma.
\end{proof}


\section{Proof of Lemma~\ref{l:mainlemma}}\label{ss:mainlemma}


\begin{proof}

{\color{black}{Our proof is similar to that of \cite[Lemma~1]{wen2015sharp}.}}
We first give some definitions.
%
%
For any given $\alpha > 0$, we define
$
\beta :=\frac{{\color{black}{1 - \sqrt{\alpha + 1}}}}{\sqrt{\alpha }}.
$
By some basic calculations, we have
  \begin{eqnarray}
  \label{e:alphaproperty}
 \frac{2 \beta }{1 - \beta ^2} = -\sqrt{\alpha }~~\text{and}~~ \label{e:alphaproperty2}
 \frac{1 + \beta ^2}{1 - \beta ^2} = \sqrt{\alpha + 1}.
 \end{eqnarray}
 Also, we define
  \begin{eqnarray}
    \label{e:B}
    \mathbf{\Phi} &\hspace{-2mm} :=&\frac{\mathbf{P}^{\bot}_{T^k}}{\sqrt{1 - \beta^4}}
    \left[ \begin{array}{c c}
    \mathbf{A}_{T \backslash {T^k}}&\mathbf{A}_{j}
    \end{array} \right],  \\
    \label{eq:u}
    \mathbf{u} \hspace{.25mm} &\hspace{-2mm} :=&\hspace{-2mm} 
    \left[ \begin{array}{c}
      \mathbf{x}_{T \backslash T^k} \\
      0
    \end{array} \right],
    \label{eq:vvector}
    \mathbf{v} = \left[ \begin{array}{c}
      \mathbf{0}_{|T \backslash T^k| \times 1}
      \\
     \lambda \beta \| \mathbf{x}_{T \backslash T^k} \|_2
    \end{array} \right],~~~~~~~~~~~
\\
  \label{e:t}
\text{where}~~\lambda \hspace{.25mm} &\hspace{-2mm} :=&\hspace{-2mm}
\begin{cases}
~1                        & \text{if}~ \langle \mathbf{A}_{j}, \mathbf{r}^{k} \rangle \geq 0, \\
-1                        & \text{if}~ \langle \mathbf{A}_{j}, \mathbf{r}^{k} \rangle < 0.
\end{cases}
\end{eqnarray}

Then, we observe that
  \begin{eqnarray}
    \lefteqn{\hspace{-5mm}\| \mathbf{\Phi} (\mathbf{u} + \mathbf{v})\|_2^2 - \| \mathbf{\Phi}
    (\beta ^2 \mathbf{u} - \mathbf{v})\|_2^2~~~~}   \nonumber\\
    & = & \hspace{-2mm} \| \mathbf{\Phi u} \|_2^2 + \| \mathbf{\Phi v} \|_2^2 + 2  \langle \mathbf{\Phi  u},  \mathbf{\Phi v} \rangle \nonumber \\
    & - & \hspace{-2mm}  \left( \| \beta ^2 \mathbf{\Phi u} \|_2^2 + \| \mathbf{\Phi  v} \|_2^2 - 2  \langle \beta ^2 \mathbf{\Phi  u},  \mathbf{\Phi  v} \rangle \right) \nonumber \\
    & = & \hspace{-2mm} (1-\beta^4)\| \mathbf{\Phi u} \|_2^2+2(1+\beta^2) \langle\mathbf{\Phi  u},  \mathbf{\Phi  v}  \rangle.
    \label{eq:cac}
      \end{eqnarray}
     %

%

\begin{figure}[t]
\centering
{\includegraphics[width = 76 mm]{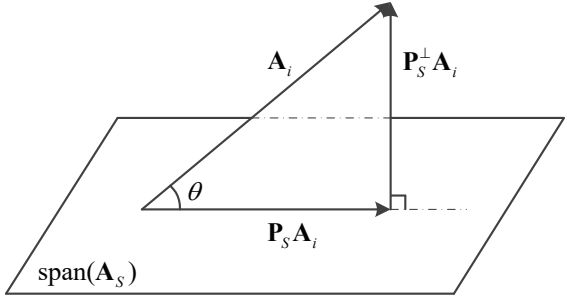}}
\caption{Geometric illustration of $\theta$ and $\mathbf{P}^{\bot}_{S}
    \mathbf{A}_{i}$.} \vspace{-5mm}\label{fig:projection}
\end{figure}

From \eqref{e:B}--\eqref{eq:vvector}, it is straightforward to show that
\begin{eqnarray}
  \label{e:AB}
    \mathbf{\Phi \mathbf{u}} &\hspace{-2mm} =&\hspace{-2mm} \frac{\mathbf{P}^{\bot}_{T^k} \mathbf{A}_{T \backslash T^k} \mathbf{x}_{T \backslash T^k}}{\sqrt{1 - \beta^4}} \overset{(a)}{=} \frac{\mathbf{r}^k}{\sqrt{1 - \beta^4}} \\
    \mathbf{\Phi \mathbf{v}} &\hspace{-2mm} =&\hspace{-2mm} \frac{\left( \lambda \beta \| \mathbf{x}_{T \backslash T^k} \|_2 \right) \mathbf{P}^{\bot}_{T^k} \mathbf{A}_{j} }{\sqrt{1 - \beta^4}},
\end{eqnarray}
where (a) is because
\begin{eqnarray}
\mathbf{r}^k &\hspace{-2mm} \overset{\eqref{eq:rkex}}{=}&\hspace{-2mm} \mathbf{P}^{\bot}_{T^{k}} \mathbf{y} = \mathbf{P}^{\bot}_{T^{k}} \mathbf{A x} = \mathbf{P}^{\bot}_{T^{k}} \left( \mathbf{A}_{T^k} \mathbf{x}_{T^k} + \mathbf{A}_{T \backslash T^k} \mathbf{x}_{T \backslash T^k} \right)  \nonumber \\
&\hspace{-2mm} =&\hspace{-2mm} \mathbf{P}^{\bot}_{T^{k}} \mathbf{A}_{T \backslash T^k} \mathbf{x}_{T \backslash T^k}. \label{eq:rk2}
\end{eqnarray}

Hence,
       \begin{eqnarray}
    \langle \mathbf{\Phi  u},  \mathbf{\Phi  v} \rangle
    &\hspace{-3mm} {=} &\hspace{-3mm} \frac{\left( \lambda \beta \| \mathbf{x}_{T \backslash T^k} \|_2 \right)
    \mathbf{A}'_{j} \left(\mathbf{P}^{\bot}_{T^k} \right)' \mathbf{P}^{\bot}_{T^k}
    \mathbf{A}_{T \backslash T^k} \mathbf{x}_{T \backslash T^k}}{1 - \beta^4} \nonumber \\
    &\hspace{-3mm} \overset{\eqref{eq:pbot}}{=} & \hspace{-3mm} \frac{\left( \lambda \beta \| \mathbf{x}_{T \backslash T^k} \|_2 \right) \mathbf{A}'_{j}
    \mathbf{P}^{\bot}_{T^k} \mathbf{A}_{T \backslash T^k} \mathbf{x}_{T \backslash T^k}}{1 - \beta^4} \nonumber\\
    & \hspace{-3mm} \overset{\eqref{eq:rk2}}{=} & \hspace{-3mm} \frac{\left( \lambda \beta  \| \mathbf{x}_{T \backslash T^k} \|_2 \right) \langle \mathbf{A}_{j}, \mathbf{r}^{k} \rangle}{1 - \beta^4} \nonumber\\
    & \hspace{-3mm} \overset{\eqref{e:t}}{=} & \hspace{-3mm} \frac{\beta \| \mathbf{x}_{T \backslash T^k} \|_2 \left| \langle \mathbf{A}_{j}, \mathbf{r}^{k} \rangle \right|}{1 - \beta^4}. \label{eq:yeor}
  \end{eqnarray}
  Using \eqref{e:AB}--\eqref{eq:yeor}, we can rewrite \eqref{eq:cac} as
    \begin{eqnarray}
\lefteqn{\hspace{-10mm}\| \mathbf{\Phi} (\mathbf{u} + \mathbf{v})\|_2^2 - \| \mathbf{\Phi} (\beta ^2 \mathbf{u} - \mathbf{v})\|_2^2} \nonumber \\
&\hspace{-3mm}  \overset{\eqref{e:AB}}{=}& \hspace{-3mm} \| \mathbf{r}^k \|_2^2
+2(1+\beta^2)\langle\mathbf{\Phi  u},  \mathbf{\Phi  v}  \rangle \nonumber \\
&\hspace{-3mm}  \overset{\eqref{eq:yeor}}{=}& \hspace{-3mm} \| \mathbf{r}^k \|_2^2
+\frac{2\beta \| \mathbf{x}_{T \backslash T^k} \|_2 \left| \langle \mathbf{A}_{j}, \mathbf{r}^{k} \rangle \right|}{1 - \beta^2} \nonumber \\
&\hspace{-3mm} \overset{\eqref{e:alphaproperty}}{=}&\hspace{-3mm}   \| \mathbf{r}^k \|_2^2
-\sqrt{\alpha}~\| \mathbf{x}_{T \backslash T^k} \|_2 \left| \langle \mathbf{A}_{j}, \mathbf{r}^{k} \rangle \right|.\label{eq:jijia}
\end{eqnarray}

 On the other hand, since both $\mathbf{u} + \mathbf{v}$ and $\beta ^2 \mathbf{u} - \mathbf{v}$ are $(|T\backslash T^k| + 1)$-sparse, and also noting that $|T \backslash T^k| = K - k$, it follows from \eqref{eq:cac} that
  \begin{eqnarray}
    \lefteqn{\hspace{-1.25mm} \| \mathbf{\Phi} (\mathbf{u} + \mathbf{v})\|_2^2 - \| \mathbf{\Phi}
    (\beta ^2 \mathbf{u} - \mathbf{v})\|_2^2} &  &  \nonumber\\
    &\hspace{-5mm} \overset{\text{Lemma}~\ref{l:orthogonalcomp}}{\geq} &\hspace{-3mm}  \frac{(1 \hspace{-.5mm} - \hspace{-.25mm} \delta_{K - k + 1}) \hspace{-.25mm} \|\mathbf{u} \hspace{-.25mm} + \hspace{-.25mm}
    \mathbf{v}\|_2^2}{1 - \beta^4} \hspace{-.5mm} - \hspace{-.5mm} \frac{(1 \hspace{-.5mm} + \hspace{-.25mm} \delta_{K - k + 1}) \hspace{-.25mm} \|(\beta ^2 \mathbf{u} \hspace{-.25mm} - \hspace{-.25mm}
    \mathbf{v})\|_2^2}{1 - \beta^4} \nonumber\\
    &\hspace{-5mm} \overset{\eqref{eq:u},\hspace{.5mm} \eqref{eq:vvector}}{=} &\hspace{-3mm} \frac{(1 \hspace{-.25mm} - \hspace{-.25mm} \delta_{K - k + 1}) (1 + \beta ^2) \|
    \mathbf{x}_{T \backslash T^k} \|_2^2}{1 - \beta^4} \nonumber\\
    &  &\hspace{-3mm} - \frac{(1 + \delta_{K - k + 1}) \beta ^2  (1 + \beta ^2) \| \mathbf{x}_{T
    \backslash T^k} \|_2^2}{1 - \beta^4} \nonumber\\
    &\hspace{-5mm}  = &\hspace{-3mm}   \left( 1 - \frac{ 1 + \beta ^2 }{1 - \beta ^2}  \delta_{K - k + 1}  \right) \| \mathbf{x}_{T \backslash T^k} \|_2^2 \nonumber\\
    &\hspace{-5mm}  \overset{\eqref{e:alphaproperty2}}{=} &\hspace{-3mm}
    \left(\hspace{-.25mm} 1 \hspace{-.25mm} - \hspace{-.25mm} \delta_{K - k + 1} \sqrt{\alpha + 1}  \right) \hspace{-.5mm} \| \mathbf{x}_{T \backslash T^k} \|_2^2 \label{eq:usa}   \overset{\eqref{inq:deltaalpha}}{>}   0. \label{eq:lasteq}
\end{eqnarray}

Combining \eqref{eq:jijia} and \eqref{eq:lasteq} completes the proof.
\end{proof}

\section{Proof of Proposition~\ref{p:mainleft}}\label{ss:mainleft}

\begin{proof}
We first prove \eqref{e:mainleft}. By \eqref{eq:rk2}, one obtains
\begin{eqnarray*}
\max_{i \in T \backslash T^k} \hspace{-.5mm} \frac{|\langle \mathbf{A}_i, \mathbf{r}^{k} \rangle|}{\| \mathbf{P}^{\bot}_{T^{k}} \mathbf{A}_i \|_2}
&\hspace{-2mm}  = &\hspace{-2mm}  \max_{i \in T \backslash T^k} \frac{|\langle \mathbf{A}_i, \mathbf{P}^{\bot}_{T^{k}} \mathbf{A}_{T \backslash T^k} \mathbf{x}_{T \backslash T^k} \rangle|}{\| \mathbf{P}^{\bot}_{T^{k}} \mathbf{A}_i \|_2} \nonumber \\
&\hspace{-2mm}  \overset{\eqref{eq:unitnorm}}{\geq}&\hspace{-2mm}   \max_{i \in T \backslash T^k} \hspace{-.25mm} {|\langle \mathbf{A}_i, \mathbf{P}^{\bot}_{T^{k}} \mathbf{A}_{T \backslash T^k} \mathbf{x}_{T \backslash T^k} \rangle|}. ~~~~~~
\end{eqnarray*}
Thus, to show \eqref{e:mainleft}, it suffices to show
\begin{eqnarray}
\max_{i \in T \backslash T^k}  {|\langle \mathbf{A}_i, \mathbf{P}^{\bot}_{T^{k}} \mathbf{A}_{T \backslash T^k} \mathbf{x}_{T \backslash T^k} \rangle|}
 \geq \frac{\| \mathbf{r}^{k} \|_2^2}{\sqrt{|T \backslash T^k|} \| \mathbf{x}_{T \backslash T^k} \|_2}.~\label{e:lbd}
\end{eqnarray}

We prove \eqref{e:lbd} by following the approach in {\cite[Lemma~1]{wen2015sharp}}.
{\color{black}{It is not hard to check that}}
  \begin{eqnarray*}
  ~~~  \lefteqn{\hspace{-6mm}\sqrt{|T \backslash T^k|} ~ \| \mathbf{x}_{T \backslash T^k} \|_2
    \max_{i \in T \backslash T^k} {|\langle \mathbf{A}_i, \mathbf{P}^{\bot}_{T^{k}} \mathbf{A}_{T \backslash T^k} \mathbf{x}_{T \backslash T^k} \rangle|}}  \nonumber\\
    & \overset{(a)}{\geq} &\hspace{-2mm}  \| \mathbf{x}_{T \backslash T^k} \|_1 \hspace{-.3mm}  \max_{i \in T \backslash T^k} {|\langle \mathbf{A}_i, \mathbf{P}^{\bot}_{T^{k}} \mathbf{A}_{T \backslash T^k} \mathbf{x}_{T \backslash T^k} \rangle|} \\
    & = &\hspace{-3mm} \hspace{-.5mm} \Bigg(\hspace{-.25mm}  \sum_{\ell \in T \backslash T^k} | x_\ell | \hspace{-.25mm}  \Bigg) \hspace{-.75mm}  \max_{i \in T \backslash T^k} \hspace{-.5mm}  \big|
    \mathbf{A}'_i \mathbf{P}^{\bot}_{T^k} \mathbf{A}_{T \backslash T^k}
    \mathbf{x}_{T \backslash T^k} \big| \\
    & \overset{(b)}{\geq} &\hspace{-3mm}  \sum_{\ell \in T \backslash T^k} \big| x_\ell \mathbf{A}'_\ell
    \mathbf{P}^{\bot}_{T^k} \mathbf{A}_{T \backslash T^k} \mathbf{x}_{T \backslash T^k}
    \big| \\
    & \geq &\hspace{-3mm}  \sum_{\ell \in T \backslash T^k} \left( x_\ell \mathbf{A}'_\ell
    \mathbf{P}^{\bot}_{T^k} \mathbf{A}_{T \backslash T^k} \mathbf{x}_{T \backslash T^k}
    \right) \\
    & = &\hspace{-3mm}   \big(\mathbf{A}_{T \backslash T^k} \mathbf{x}_{T \backslash T^k}\big)'
    \mathbf{P}^{\bot}_{T^k} \mathbf{A}_{T \backslash T^k} \mathbf{x}_{T \backslash T^k}
    \\
    & \overset{\eqref{eq:pbot}}{=}
     &\hspace{-2mm}  \| \mathbf{P}^{\bot}_{T^k} \mathbf{A}_{T \backslash T^k} \mathbf{x}_{T
    \backslash T^k} \|_2^2
    \overset{\eqref{eq:rk2}}{=}  \|\mathbf{r}^{k} \|_2^2,
  \end{eqnarray*}
  where (a) follows from Cauchy-Schwarz inequality and (b) is because for each $\ell \in
  T \backslash T^k$,
\[
    \max_{i \in T \backslash T^k} | \mathbf{A}'_i \mathbf{P}^{\bot}_{T^k}
     \mathbf{A}_{T \backslash T^k} \mathbf{x}_{T \backslash T^k} |\geq  |\mathbf{A}'_\ell \mathbf{P}^{\bot}_{T^k} \mathbf{A}_{T \backslash T^k} \mathbf{x}_{T
     \backslash T^k}|.~
\]
Thus \eqref{e:lbd} holds, and this completes the proof of \eqref{e:mainleft}.


We next move to the proof of \eqref{e:mainright}, Let
  $
  j_0 := \underset{i \in \{1, \cdots, n\} \backslash T}{\arg \max} \frac{|\langle \mathbf{A}_i, \mathbf{r}^{k} \rangle|}{\| \mathbf{P}^{\bot}_{T^{k}} \mathbf{A}_i \|_2}.
  $
Proving (\ref{e:mainright}) is equivalent to showing 
$$
 \frac{\| \mathbf{r}^{k} \|_2^2}{\sqrt{|T \backslash T^k|} \| \mathbf{x}_{T \backslash T^k} \|_2} > \frac{|\langle \mathbf{A}_{j_0}, \mathbf{r}^{k} \rangle|}{\| \mathbf{P}^{\bot}_{T^{k}} \mathbf{A}_{j_0} \|_2}, 
$$
\begin{equation}
\hspace{-8mm} \text{i.e.,}~~ \| \mathbf{r}^{k} \|_2^2 > \frac{\sqrt{|T \backslash T^k|}  }{\| \mathbf{P}^{\bot}_{T^{k}} \mathbf{A}_{j_0} \|_2} \| \mathbf{x}_{T \backslash T^k} \|_2 |\langle \mathbf{A}_{j_0}, \mathbf{r}^{k} \rangle|. ~~\label{eq:japan}
\end{equation}

By applying Lemma~\ref{l:mainlemma} with
$\alpha=\frac{|T \backslash T^k|}{\| \mathbf{P}^{\bot}_{T^{k}} \mathbf{A}_{j_0} \|_2^2},$\footnote{In this way, the column-normalization feature of OLS is incorporated. We note here that $\alpha$ may not be an integer.}
we can see that \eqref{eq:japan} holds if
\begin{equation}
\delta_{K - k + 1}<1/\sqrt{\frac{|T \backslash T^k|}{\| \mathbf{P}^{\bot}_{T^{k}} \mathbf{A}_{j_0} \|_2^2}+1}. \label{eq:china}
\end{equation}

In fact, applying Lemma~\ref{l:monot} yields {\color{black}{$\delta_{K+1} \geq \delta_{K - k + 1}$}}, thus~\eqref{eq:china} is guaranteed by \eqref{eq:o} whenever
  \begin{equation}
    \label{inq:TSK}
    \frac{|T \backslash T^k|}{\| \mathbf{P}^{\bot}_{T^k} \mathbf{A}_{j_0} \|_2^2} \leq K.
  \end{equation}
Thus, to prove \eqref{eq:japan}, it suffices to show that \eqref{inq:TSK} holds for all $k \in \{0, \cdots, K - 1\}$. To complete the argument, we consider the following two cases.
First, if $k = 0$,  we immediately have
\[
  \frac{|T \backslash T^0|}{\| \mathbf{P}^{\bot}_{T^0} \mathbf{A}_{j_0} \|_2^2} = \frac{|T|}{\| \mathbf{A}_{j_0} \|_2^2} \overset{\eqref{eq:unitnorm}}{\leq} K.~~
  \]

Second, if $1 \leq k \leq K - 1$, then $|T \backslash T^k| \leq K - 1$. Hence,
{\color{black}{
\begin{eqnarray} \label{eq:maomi}
  \frac{|T \backslash T^k|}{\| \mathbf{P}^{\bot}_{T^k} \mathbf{A}_{j_0} \|_2^2} 
  &\hspace{-2mm} \overset{\text{Lemma}~\ref{l:projld}}{\leq}&\hspace{-2mm} \frac{K - 1}{1 - \delta^2_{|T^k| + 1}} ~=~ \frac{K - 1}{1 - \delta^2_{k + 1}} \nonumber \\
&\hspace{-2mm} \overset{\text{Lemma}~\ref{l:monot}}{\leq}&\hspace{-2mm} \frac{K - 1}{1 - \delta^2_{K + 1}}
~~~\hspace{-.3mm}\overset{(\ref{eq:o})}{\leq}~ \frac{K - 1}{1 - \frac{1}{K + 1}} \leq K.~~~~~~~
\end{eqnarray}
}}
Therefore, \eqref{inq:TSK} holds and the proof is complete.
\end{proof}

{\color{black}{

\section{Main Steps for Constructing Example~\ref{ex:countex}}\label{app:example}

\newcounter{mytempeqcnt21}
\setcounter{mytempeqcnt21}{\value{equation}}
\setcounter{equation}{4}
\begin{figure*}[h!]
\normalsize 
\begin{eqnarray}\label{eq:adag}
\lambda_1  &\hspace{-2.5mm}=&\hspace{-2.5mm} 1/3\,\left( \left( 3\,\left( -24\,{a}^{6}+ \left( 72\,b+72 \right) {
a}^{5}+ \left( -63\,{b}^{2}-216\,b-108 \right) {a}^{4}+ \left( 6\,{b}^
{3}+198\,{b}^{2}+288\,b+96 \right) {a}^{3}+ \left( -63\,{b}^{4}-198\,{
b}^{3} \right. \right. \right. \right. \nonumber \\
&&\hspace{-4mm} \left. \left. \left. \left. -315\,{b}^{2}-216\,b-54 \right) {a}^{2}+72\, \left( b+1 \right) 
 \left( {b}^{2}+b+1/2 \right) ^{2}a-24\, \left( {b}^{2}+b+1/2 \right) 
^{3} \right)^{1/2}
-27\,{a}^{2}b+ \left( 27\,{b}^{2}+27\,b \right) a \right) ^{2/3}
\right. \nonumber \\
&&\hspace{-3mm} \left. +
3\,\left( 3\,\left( -24\,{a}^{6}+ \left( 72\,b+72 \right) {a}^{5}+
 \left( -63\,{b}^{2}-216\,b-108 \right) {a}^{4}+ \left( 6\,{b}^{3}+198
\,{b}^{2}+288\,b+96 \right) {a}^{3}+ \left( -63\,{b}^{4}-198\,{b}^{3} \right.  \right. \right. \right. \nonumber \\
&&\hspace{-4mm} \left. \left. \left. \left. -315\,{b}^{2}-216\,b-54 \right) {a}^{2}+72\, \left( b+1 \right) 
 \left( {b}^{2}+b+1/2 \right) ^{2}a-24\, \left( {b}^{2}+b+1/2 \right) 
^{3} \right)^{1/2}
-27\,{a}^{2}b+ \left( 27\,{b}^{2}+27\,b \right) a \right)^{1/2}  \right. \nonumber \\
&&\hspace{-3mm} \left.
+6\,{a}^{2}+
 \left( -6\,b-6 \right) a+6\,{b}^{2}+6\,b+3 \right)   \left( 3\,\left( -24\,
{a}^{6}+ \left( 72\,b+72 \right) {a}^{5}+ \left( -63\,{b}^{2}-216\,b-
108 \right) {a}^{4}+ \left( 6\,{b}^{3}+198\,{b}^{2}
\right. \right. \right. \nonumber \\
&&\hspace{-3.75mm} \left. \left. \left. +288\,b+96 \right) 
{a}^{3}+ \left( -63\,{b}^{4}-198\,{b}^{3}-315\,{b}^{2}-216\,b-54
 \right) {a}^{2}+72\, \left( b+1 \right)  \left( {b}^{2}+b+1/2
 \right) ^{2}a-24\, \left( {b}^{2}+b+1/2 \right) ^{3} \right)^{1/2} \right.  \nonumber \hspace{-4mm} \\
&&\hspace{-3mm} \left. 
 -27\,{a}^{2}b+
 \left( 27\,{b}^{2}+27\,b \right) a \right)^{-1/3},
\end{eqnarray}

\vspace{-6mm}
\begin{eqnarray}
\lambda_2 &\hspace{-2.5mm}=&\hspace{-2.5mm}  - \left( \left( -i/ 6\sqrt {3}+1/6 \right) 
 \left( 3\,\left( -24\,{a}^
{6}+ \left( 72\,b+72 \right) {a}^{5}+ \left( -63\,{b}^{2}-216\,b-108
 \right) {a}^{4}+ \left( 6\,{b}^{3}+198\,{b}^{2}+288\,b+96 \right) {a}
^{3} 
\right. \right. \right. \nonumber \hspace{-4mm} \\
&&\hspace{-3.5mm} \left. \left. \left. 
+ \left( -63\,{b}^{4}-198\,{b}^{3}-315\,{b}^{2}-216\,b-54 \right) 
{a}^{2}+72\, \left( b+1 \right)  \left( {b}^{2}+b+1/2 \right) ^{2}a-24
\, \left( {b}^{2}+b+1/2 \right) ^{3} \right)^{1/2}
-27\,{a}^{2}b 
\right. \right.  \nonumber \hspace{-4mm} \\
&&\hspace{-3mm} \left. \left.  
+ \left( 27\,{b}^{2}
+27\,b \right) a \right) ^{2/3}- \left( 3\,\left( -24\,{a}^{6}+
 \left( 72\,b+72 \right) {a}^{5}+ \left( -63\,{b}^{2}-216\,b-108
 \right) {a}^{4}+ \left( 6\,{b}^{3}+198\,{b}^{2}+288\,b+96 \right) {a}
^{3}
\right. \right. \right. \nonumber \hspace{-4mm} \\
&&\hspace{-3.5mm} \left. \left. \left. 
+ \left( -63\,{b}^{4}-198\,{b}^{3}-315\,{b}^{2}-216\,b-54 \right) 
{a}^{2}+72\, \left( b+1 \right)  \left( {b}^{2}+b+1/2 \right) ^{2}a-24
\, \left( {b}^{2}+b+1/2 \right) ^{3} \right)^{1/2}
-27\,{a}^{2}b 
\right. \right. \nonumber \hspace{-4mm} \\
&&\hspace{-3mm} \left. \left. 
+ \left( 27\,{b}^{2}
+27\,b \right) a \right)^{1/3}
+ \left( i{a}^{2}+ \left( -ib-i \right) a+i{b}^{2}+ib
+i/2 \right) \sqrt {3}+{a}^{2}+ \left( -b-1 \right) a+{b}^{2}+b+1/2 \right)
\left( 3\,\left( -24\,{a}^{6} 
\right. \right.  \nonumber \hspace{-4mm} \\
&&\hspace{-3mm} \left. \left.  
+ \left( 72\,b+72 \right) {a}^{5}+
 \left( -63\,{b}^{2}-216\,b-108 \right) {a}^{4}+ \left( 6\,{b}^{3}+198
\,{b}^{2}+288\,b+96 \right) {a}^{3}+ \left( -63\,{b}^{4}-198\,{b}^{3}-
315\,{b}^{2}-216\,b 
\right. \right. \right. \nonumber \hspace{-4mm} \\
&&\hspace{-3mm} \left. \left. \left. 
-54 \right) {a}^{2}+72\, \left( b+1 \right) 
 \left( {b}^{2}+b+1/2 \right) ^{2}a-24\, \left( {b}^{2}+b+1/2 \right) 
^{3} \right)^{1/2}
-27\,{a}^{2}b+ \left( 27\,{b}^{2}+27\,b \right) a \right)^{-1/3}
\end{eqnarray}

\vspace{-6mm}
\begin{eqnarray}\label{eq:adag3}
\lambda_3 &\hspace{-2.5mm}=&\hspace{-2.5mm}  \left( \left( -i/ 6\sqrt {3} - 1/6 \right) 
 \left( 3\,\left( -24\,{a}^
{6}+ \left( 72\,b+72 \right) {a}^{5}+ \left( -63\,{b}^{2}-216\,b-108
 \right) {a}^{4}+ \left( 6\,{b}^{3}+198\,{b}^{2}+288\,b+96 \right) {a}
^{3} 
\right. \right. \right. \nonumber \hspace{-4mm} \\
&&\hspace{-3.5mm} \left. \left. \left. 
+ \left( -63\,{b}^{4}-198\,{b}^{3}-315\,{b}^{2}-216\,b-54 \right) 
{a}^{2}+72\, \left( b+1 \right)  \left( {b}^{2}+b+1/2 \right) ^{2}a-24
\, \left( {b}^{2}+b+1/2 \right) ^{3} \right)^{1/2}
-27\,{a}^{2}b 
\right. \right.  \nonumber \hspace{-4mm} \\
&&\hspace{-3mm} \left. \left.  
+ \left( 27\,{b}^{2}
+27\,b \right) a \right) ^{2/3} + \left( 3\,\left( -24\,{a}^{6}+
 \left( 72\,b+72 \right) {a}^{5}+ \left( -63\,{b}^{2}-216\,b-108
 \right) {a}^{4}+ \left( 6\,{b}^{3}+198\,{b}^{2}+288\,b+96 \right) {a}
^{3}
\right. \right. \right. \nonumber \hspace{-4mm} \\
&&\hspace{-3.5mm} \left. \left. \left. 
+ \left( -63\,{b}^{4}-198\,{b}^{3}-315\,{b}^{2}-216\,b-54 \right) 
{a}^{2}+72\, \left( b+1 \right)  \left( {b}^{2}+b+1/2 \right) ^{2}a-24
\, \left( {b}^{2}+b+1/2 \right) ^{3} \right)^{1/2}
-27\,{a}^{2}b 
\right. \right. \nonumber \hspace{-4mm} \\
&&\hspace{-3mm} \left. \left. 
+ \left( 27\,{b}^{2}
+27\,b \right) a \right)^{1/3}
+ \left( i{a}^{2}+ \left( -ib-i \right) a+i{b}^{2}+ib
+i/2 \right) \sqrt {3}+{a}^{2}+ \left( -b-1 \right) a+{b}^{2}+b+1/2 \right)
\left( 3\,\left( -24\,{a}^{6} 
\right. \right.  \nonumber \hspace{-4mm} \\
&&\hspace{-3mm} \left. \left.  
+ \left( 72\,b+72 \right) {a}^{5}+
 \left( -63\,{b}^{2}-216\,b-108 \right) {a}^{4}+ \left( 6\,{b}^{3}+198
\,{b}^{2}+288\,b+96 \right) {a}^{3}+ \left( -63\,{b}^{4}-198\,{b}^{3}-
315\,{b}^{2}-216\,b 
\right. \right. \right. \nonumber \hspace{-4mm} \\
&&\hspace{-3mm} \left. \left. \left. 
-54 \right) {a}^{2}+72\, \left( b+1 \right) 
 \left( {b}^{2}+b+1/2 \right) ^{2}a-24\, \left( {b}^{2}+b+1/2 \right) 
^{3} \right)^{1/2}
-27\,{a}^{2}b+ \left( 27\,{b}^{2}+27\,b \right) a \right)^{-1/3}
\end{eqnarray}
\vspace{-4mm}

\hrulefill
\end{figure*}
\newcounter{mytempeqcnt22}
\setcounter{mytempeqcnt22}{\value{equation}}
\setcounter{equation}{0}

We aim to construct a $K$-sparse signal $\mathbf{x}$ and a matrix $\mathbf{A}$ satisfying the RIP with {\color{black}{$\delta_{K + 1} = \frac{1}{\sqrt{K + \beta}}$}} where $\beta \in (0,1]$ is a constant, for which OLS fails to recover $\mathbf{x}$ in $K$ steps. Here, we wish $\beta$ to be as close to one as possible so as to minimize the gap to the sufficient condition $\delta_{K + 1} < \frac{1}{\sqrt{K + 1}}$. The main steps for maximizing $\beta$ are as follows:
    \begin{enumerate}[i)]

  \item For simplicity, in the following we construct a $2$-sparse signal $\mathbf{x} \in \mathbb{R}^3$ and a $3$-by-$3$ symmetric matrix $\mathbf{A}'\mathbf{A}$: 
    \begin{equation} \label{eq:hi}
    \mathbf{x} = \left[ \begin{array}{c}
       \hspace{-.5mm} 0 \hspace{-.5mm}\\
       \hspace{-.5mm} - 1 \hspace{-.5mm}\\
       \hspace{-.5mm} 1 \hspace{-.5mm}
    \end{array} \right]~\text{and}~\mathbf{A}'\mathbf{A} =
    \left[ \begin{array}{ccc}
      1            &      a        & b \\
      a            &      1        & c \\
      b            &      c        & 1
    \end{array} \right], 
    \end{equation}
    where $a,b,c \in (-1,1)$ are parameters to be determined.\footnote{{\color{black}{We set the diagonal values of $\mathbf{A}'\mathbf{A}$ to ones to ensures that $\mathbf{A}$ has unit $\ell_2$-norm columns. Then, we must have $a,b,c \neq \pm 1$ since otherwise the columns of $\mathbf{A}$ would be linearly dependent and also the RIP would not be satisfied (e.g., $b = \langle \mathbf{A}_1, \mathbf{A}_3 \rangle  = 1$ implies $\mathbf{A}_1 =  \mathbf{A}_3$).}}} Then, we try to find out the optimal set of $\{a, b, c\}$ that maximizes $\beta$ in the isometry constant \begin{equation} \label{eq:deng}
    \delta_3 = \left.\frac{1}{\sqrt{K + \beta}} \right|_{K = 2}= \frac{1}{\sqrt{2 + \beta}}.
    \end{equation}

\item To simplify the problem, we explore the relationship between $a$, $b$ and $c$. 
Observe that 
\begin{eqnarray}
\left| \langle \mathbf{A}_1, \mathbf{y} \rangle \right| &\hspace{-2mm}=&\hspace{-2mm}  |\langle (\mathbf{A}'\mathbf{A})_1, \mathbf{x} \rangle| = |a - b|, \nonumber\\
\left| \langle \mathbf{A}_2, \mathbf{y} \rangle \right| &\hspace{-2mm}=&\hspace{-2mm}  |\langle (\mathbf{A}'\mathbf{A})_2, \mathbf{x} \rangle| = 1 - c,  \nonumber \\
\left| \langle \mathbf{A}_3, \mathbf{y} \rangle \right| &\hspace{-2mm}=&\hspace{-2mm} |\langle (\mathbf{A}'\mathbf{A})_3, \mathbf{x} \rangle|  = 1 - c. \nonumber
\end{eqnarray}
We consider the case where $\left| \langle \mathbf{A}_1, \mathbf{y} \rangle \right| = \left| \langle \mathbf{A}_2, \mathbf{y} \rangle \right| =\left| \langle \mathbf{A}_3, \mathbf{y} \rangle \right|$, that is, three columns of $\mathbf{A}$ have the same correlation with $\mathbf{y}$. By the tie-breaking rule in Table~\ref{a:OLS}, this case is the critical case for OLS to select a wrong index $t^1 = 1$ in the first iteration. By investigating this case, we expect to obtain the best possible value for $\beta$.
Specifically, we have 
$
|a - b| = 1 - c, 
$
and thus
\begin{equation}
c =
\begin{cases}
1 - a + b                     & a \geq b, \\ 
1 + a - b                     & a < b.
\end{cases}  
\end{equation}
Without loss of generality, we only consider the case where $a \geq b$ since permuting the matrix yields the other case. Applying $c = 1 - a + b$ in~\eqref{eq:hi}, we have 
\begin{equation}
\mathbf{A}'\mathbf{A} =
    \left[ \begin{array}{ccc}
      1            &      a                & b \\
      a            &      1                & 1 - a + b \\
      b            &      1 - a + b        & 1
    \end{array} \right] \label{eq:AA}
\end{equation}

\newcounter{mytempeqcnt31}
\setcounter{mytempeqcnt31}{\value{equation}}
\setcounter{equation}{4}

\newcounter{mytempeqcnt32}
\setcounter{mytempeqcnt32}{\value{equation}}
\setcounter{equation}{7}

The eigenvalues of $\mathbf{A}'\mathbf{A}$ are given in \eqref{eq:adag}--\eqref{eq:adag3}, which are obtained with the aid of Maple 18. One can check that $\lambda_1$, $\lambda_2$, and $\lambda_3$ are continuous when $-1 < b \leq a < 1$ and moreover, 
\begin{eqnarray}
\hspace{-1mm}\begin{cases}
\sup_{-1 < b \leq a < 1} \lambda_1 \rightarrow  3,       & \text{when}~a\rightarrow 1,~b \rightarrow -1, \\ 
\inf_{-1 < b \leq a < 1} \lambda_1 = \frac{4}{3},        & \text{when}~a=\frac{1}{3},~b =-\frac{1}{3},  \\
\sup_{-1 < b \leq a < 1} \lambda_2 = \frac{1}{3},        & \text{when}~a=\frac{1}{3},~b =-\frac{1}{3},  \\
\inf_{-1 < b \leq a < 1} \lambda_2 \rightarrow 0,        & \text{when}~a\rightarrow 1,~b \rightarrow -1,  \\
\sup_{-1 < b \leq a < 1} \lambda_3 = \frac{4}{3},        & \text{when}~a=\frac{1}{3},~b =-\frac{1}{3}, \\
\inf_{-1 < b \leq a < 1} \lambda_3 \rightarrow 0,        & \text{when}~a\rightarrow 1,~b \rightarrow -1. \label{eq:duorou}
\end{cases}   
\end{eqnarray}

%

\item Note that our goal is to maximize $\beta$, which is equivalent to minimizing $\delta_3$. Thus, by exploring the relationship between $\delta_3$ and the eigenvalues of $\mathbf{A}'\mathbf{A}$~\cite[Remark~1]{dai2009subspace}, one can show from~\eqref{eq:duorou} that 
\begin{eqnarray}
 \lefteqn{ \min_{1 < b \leq a < 1} \delta_3} \nonumber \\
&\hspace{-2mm}=&\hspace{-2mm} \min_{1 < b \leq a < 1}   \max  \left\{  \lambda_{\max}(\mathbf{A}'\mathbf{A})   -  1,  1  -   \lambda_{\min}(\mathbf{A}'\mathbf{A}) \right\}  \nonumber   \\
&\hspace{-2mm}=&\hspace{-2mm}   \max \left\{ \inf_{1 < b \leq a < 1}   \lambda_1   -   1,  1  -  \sup_{-1 < b \leq a < 1}   \lambda_2, \right. \nonumber \\
&\hspace{-2mm} &\hspace{-2mm}  \left. \min_{1 < b \leq a < 1}|\lambda_3 - 1|  \right\}  = \max \left\{ \frac{1}{3}, \frac{2}{3}, 0 \right\} = \frac{2}{3}, \hspace{-2mm} \label{eq:dajiaii}
\end{eqnarray}
and hence
\begin{equation}
 \max_{1 < b \leq a < 1} \beta \overset{\eqref{eq:deng}}{=} \left( \min_{1 < b \leq a < 1}  \delta_3 \right)^{-2} -2 \overset{\eqref{eq:dajiaii}}{=} \frac{1}{4}.
 \end{equation} 
 
 Finally, note that the maxima is attained at $(a, b)  = \left( \frac{1}{3}, - \frac{1}{3}\right)$. Applying this to \eqref{eq:AA} and performing a Cholesky decomposition on $\mathbf{A}'\mathbf{A}$ yield $\mathbf{A}$. 

    \end{enumerate}
}}

{\color{black}{
\section{Proof of Theorem~\ref{t:transferT}}\label{ss:relat}

\begin{proof}
Let $\mathbf{x}\in \mathbb{R}^{n}$ be an arbitrary $K$-sparse vector,
then $\mathbf{D x}$ is also $K$-sparse, which together with \eqref{eq:RIP} implies that
\begin{eqnarray}
\label{e:transferTub}
\| \hat{\mathbf{A}}\mathbf{D x} \|_2^2 
&\hspace{-2mm}  \leq &\hspace{-2mm} \big(1 + \delta_K(\hat{\mathbf{A}}) \big) \| \mathbf{D x} \|_2^2 \nonumber\\
&\hspace{-2mm} \leq &\hspace{-2mm} \big(1 + \delta_K (\hat{\mathbf{A}}) \big) \max_{1\leq i\leq n} d^2_{ii} \| \mathbf{x} \|_2^2 \nonumber\\
&\hspace{-2mm} = &\hspace{-2mm} \left[1+\left(\big(1 + \delta_K (\hat{\mathbf{A}}) \big) \max_{1\leq i \leq}d^2_{ii}-1\right)\right]  \| \mathbf{x} \|_2^2. ~~~~~~~
\end{eqnarray}
Similarly,
\begin{eqnarray}
\label{e:transferTlb}
\| \hat{\mathbf{A}}\mathbf{D x} \|_2^2 
&\hspace{-2mm} \geq &\hspace{-2mm} \big(1 - \delta_K(\hat{\mathbf{A}}) \big) \| \mathbf{D x} \|_2^2 \nonumber\\
&\hspace{-2mm} \geq &\hspace{-2mm} \big(1 - \delta_K (\hat{\mathbf{A}}) \big) \min_{1\leq i \leq n} d^2_{ii} \| \mathbf{x} \|_2^2 \nonumber\\
&\hspace{-2mm} = &\hspace{-2mm} \left[1-\left(1-   \big(1 - \delta_K (\hat{\mathbf{A}})  \big)   \min_{1\leq i \leq n}d_{ii}^2\right)\right]    \| \mathbf{x} \|_2^2. ~~~~~~~
\end{eqnarray}

Since $\mathbf{A} = \hat{\mathbf{A}}\mathbf{D}$, by \eqref{e:transferTub}, \eqref{e:transferTlb}, \eqref{e:gamma} and the definition of the RIP,
one can see that $\delta_K(\mathbf{A}) \leq \gamma$ holds.
\end{proof}
}}

{\color{black}{
\section{Proof of Corollary~\ref{c:transferT}}\label{app:cc}

\begin{proof}
Since $\hat{\mathbf{A}}$ satisfies the RIP of order $K$ with isometry constant $\delta_K(\hat{\mathbf{A}})$,  by Lemma \ref{l:monot}, for $i = 1, \cdots, n$,
we have
\[
1 - \delta_K (\hat{\mathbf{A}}) \leq 1 - \delta_1 (\hat{\mathbf{A}})
\leq \|\hat{\mathbf{A}}_i\|^2_2 \leq 1 + \delta_1 (\hat{\mathbf{A}})
\leq 1 + \delta_K (\hat{\mathbf{A}}).
\]
Noting that $d_{ii}=\frac{1}{\|\hat{\mathbf{A}}_i\|_2},$ we obtain
\[
\frac{1}{1 + \delta_K (\hat{\mathbf{A}})} \leq d_{ii}^2 \leq \frac{1}{1 - \delta_K (\hat{\mathbf{A}})}, ~~ i = 1, \cdots, n.
 \]
 Thus,
 \[
 (1+\delta_K(\hat{\mathbf{A}}))\max_{1\leq i\leq n}d_{ii}-1
 \leq \frac{1+\delta_K(\hat{\mathbf{A}})}{1-\delta_K(\hat{\mathbf{A}})}-1
 =\frac{2\delta_K(\hat{\mathbf{A}})}{1-\delta_K(\hat{\mathbf{A}})},
 \]
 and
 \[
 1-(1-\delta_K(\hat{\mathbf{A}}))\min_{1\leq i\leq n}d_{ii}\leq
 1-\frac{1-\delta_K(\hat{\mathbf{A}})}{1+\delta_K(\hat{\mathbf{A}})}
 =\frac{2\delta_K(\hat{\mathbf{A}})}{1+\delta_K(\hat{\mathbf{A}})}.
 \]
 Then, by $\delta_K(\mathbf{A}) \leq \gamma$ and \eqref{e:gamma}, one can easily see that \eqref{e:transferT2} holds, which completes the proof.
\end{proof}

}}

\end{document}